\documentclass{llncs}
\usepackage{graphicx}
\usepackage{theorem}
\usepackage{balance}
\usepackage{amsmath,amssymb}
\usepackage{url}
\newcommand{\R}{\mathbb R}
\newcommand{\Omit}[1]{}
\newcommand{\up}{ \uparrow}

\newcommand{\down}{\downarrow}

\newcommand{\Err}{ {\mbox{\em Err}}}
\newcommand{\Gen}{ {\mbox{\em Gen}}}

\begin{document}

 \author{
    Chengfang Fang \hspace{3.5cm}  Ee-Chien Chang
 }
 \institute{
School of Computing \\
National University of Singapore\\
    {\tt c.fang@comp.nus.edu.sg} \hspace{0.5cm} {\tt
    changec@comp.nus.edu.sg}
 }
\titlerunning{Publishing Location Dataset Differential Privately with Isotonic Regression}

\title{Publishing Location Dataset Differential Privately with Isotonic Regression}

\maketitle
\begin{abstract}
We consider the problem of publishing location datasets, in particular 2D spatial pointsets, in a differentially private manner.  Many existing mechanisms focus on frequency counts of the points in some a priori partition of the domain that is difficult to determine.
We propose an approach that adds noise directly to the  point, or to a group of neighboring points. Our approach is based on the observation that, the sensitivity of sorting, as a function on sets of real numbers, can be bounded. Together with isotonic regression, the dataset can be accurately reconstructed. To extend the mechanism to higher dimension, we employ locality preserving function to map the dataset to a bounded interval. Although there are fundamental limits on the performance of locality preserving functions, fortunately, our problem only requires distance preservation in the ``easier'' direction, and the well-known Hilbert space-filling curve suffices to provide high accuracy. The publishing process is simple from the publisher's point of view: the publisher just needs to map the data, sort them, group them, add Laplace noise and publish the dataset. The only  parameter to determine is the group size which can be chosen based on predicted generalization errors. Empirical study shows that the published dataset can also exploited to answer other queries, for example, range query and median query, accurately.
\end{abstract}

\section{Introduction}

The popularity of personal devices equipped with location sensors leads to large amount of location data being gathered. Such data contain rich information and would be valuable if they can be shared and published.  As the data may reveal location of identified individual, it is important to anonymize the data before publishing. The recently developed notion of differential privacy \cite{dwork2006differential} provides a strong form of privacy assurance regardless of the background information held by the adversaries. Such assurance is important as many case studies and past events have shown that a seemingly annoymized dataset together with additional knowledge held by the adversary could reveal information on individuals.

Most studies on differential privacy focus on publishing statistical values, for instance, k-means\cite{blum2005practical}, private coreset\cite{feldman2009private}, and median of the database\cite{nissim2007smooth}. Publishing  specific statistics or data-mining results is meaningful if the publisher knows what the public specifically want. However, there are situations where the publishers want to give the public greater flexibility in analyzing and exploring the data, for example, using different visualization  techniques. In such scenarios, it is desired to {\em ``publish data, not the data mining result''}~\cite{fung2010privacy}.

\begin{figure}[!h] \centering
\includegraphics[width=0.5\textwidth]{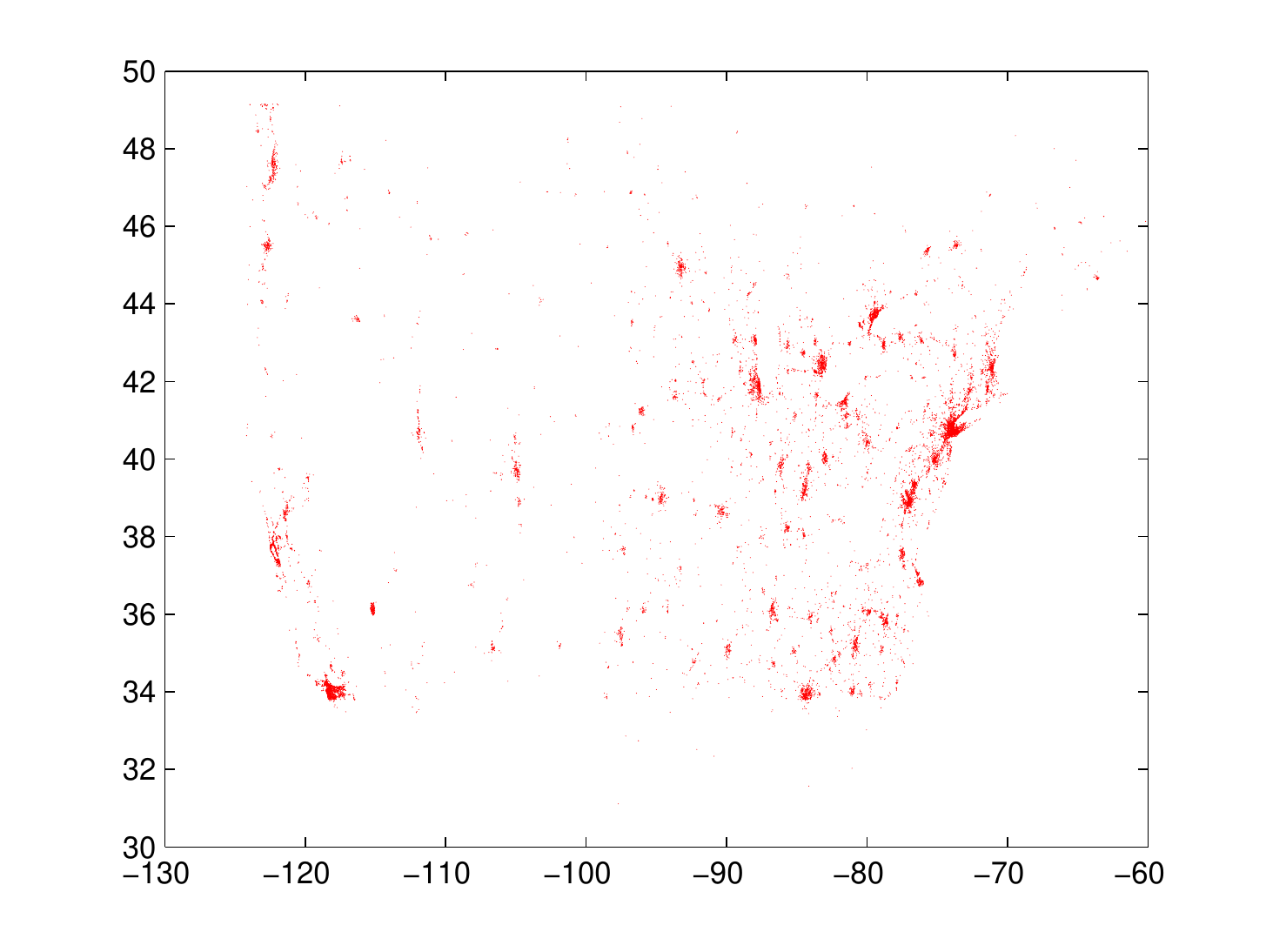}
\caption{Twitter location data cropped at the North America region.  To avoid clogging, only 10\% of the points (randomly chosen) are plotted.}
\label{fig:NAorigindata}
\end{figure}

In this paper, we consider the problem of publishing location data, or other low dimensional data in a differential private manner.  An example is shown in Fig. \ref{fig:NAorigindata} which depicts the locations of 183,072 Twitter users in North America~\cite{website}, and Fig. \ref{fig:originaldata_1d} shows a sequence of sorted real number obtained by mapping the points in Fig. \ref{fig:NAorigindata} into the unit interval. We proposed a mechanism based on the observation that sorting, as a function that takes in a set of real numbers from the unit interval, interestingly has sensitivity one (Theorem \ref{thm:sensofsort}).
Hence $\epsilon$-differential privacy can be achieved by adding Laplace noise with a scale parameter $1/\epsilon$ directly  to the sorted sequence.  Fig. \ref{fig:isotonicwithoutgroup} shows such  noisy data by adding noise to the curve in  Fig. \ref{fig:originaldata_1d}.  Although  seemingly noisy, as the original sequence is sorted,  there are dependencies among them to be exploited. Fig. \ref{fig:reconstructedwithouterror} shows a reconstructed sequence using isotonic regression.

To further  reduce perturbation induced by the Laplace noise,  consecutive elements in the sorted sequence can be grouped. However, grouping introduces {\em generalization error}.  The amount of  generalization error in the ``worst case'' can be analytically determined, and together with the model of error induced by the Laplace noise,  the publisher can  choose an appropriate group size $k$ based on the privacy requirement $\epsilon$ and the total number of points $n$. For the example in Fig. 1, the group size determined  is $300$ and the corresponding published and reconstructed data are depicted in Fig. \ref{fig:isotonicwithgroup}. Fig. \ref{fig:compare2error} shows a comparison of the error of each points in the reconstructed sequence. After reconstruction, the inverse mapping is to be applied to the data. Our variant of isotonic regression outputs dataset with larger number of repetition, as shown in Fig. \ref{fig:mapusbeforepp}.  Fig. \ref{fig:reconstructed2d} shows the reconstructed pointset, with a post-processing that maps the repeating points to the surrounding area of the location. The post-processing is solely for the purpose of visualization. Fig. \ref{fig:eastcoast} shows a zoomed version region around New York City.

\begin{figure}[!h] \centering
\includegraphics[width=0.5\textwidth]{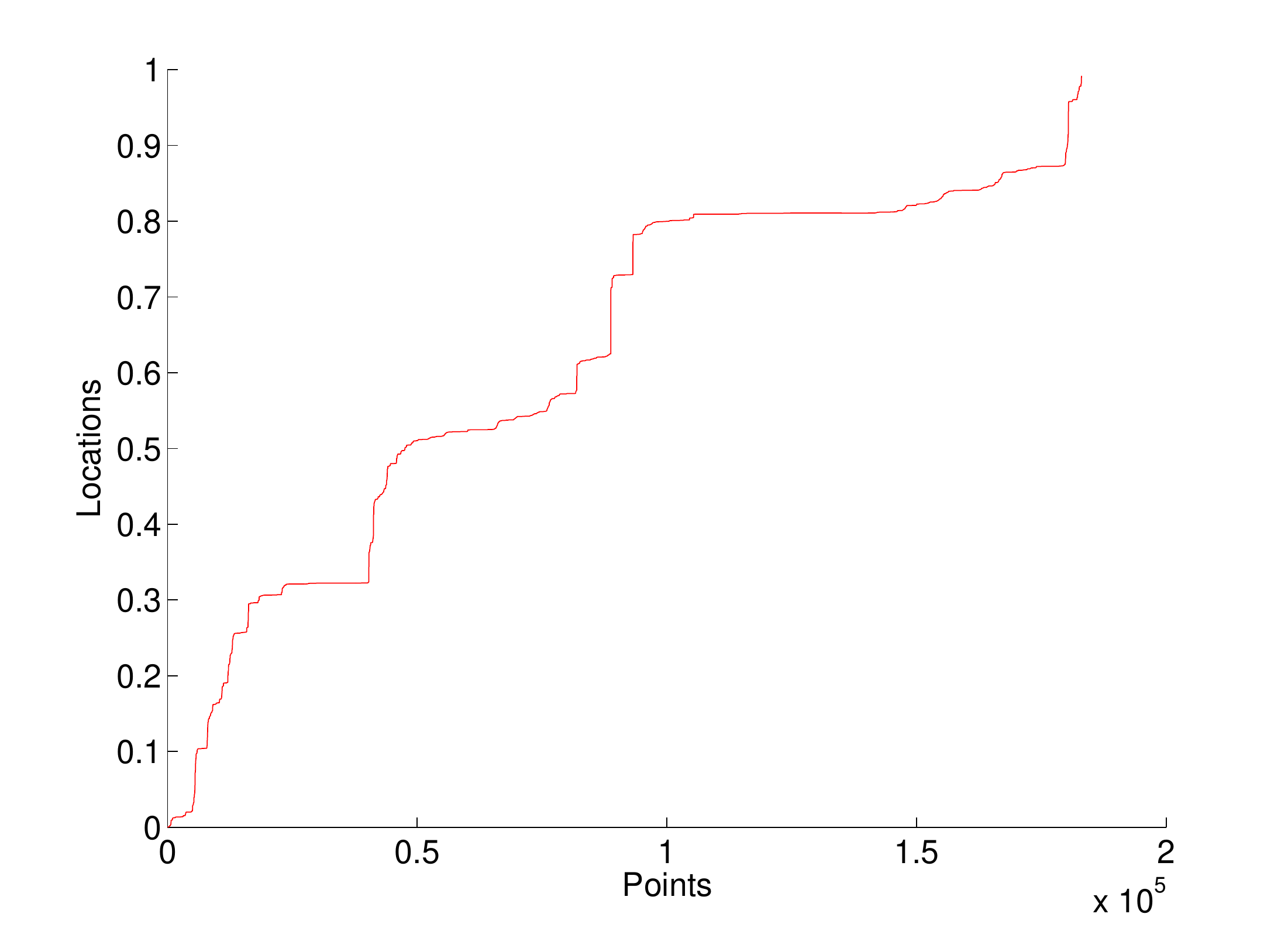}
\caption{The 1D data mapped from 2D Twitter locations (Fig. \ref{fig:NAorigindata}) to the unit interval $[0,1]$. }
\label{fig:originaldata_1d}
\end{figure}

\begin{figure}[!h] \centering
\includegraphics[width=0.5\textwidth]{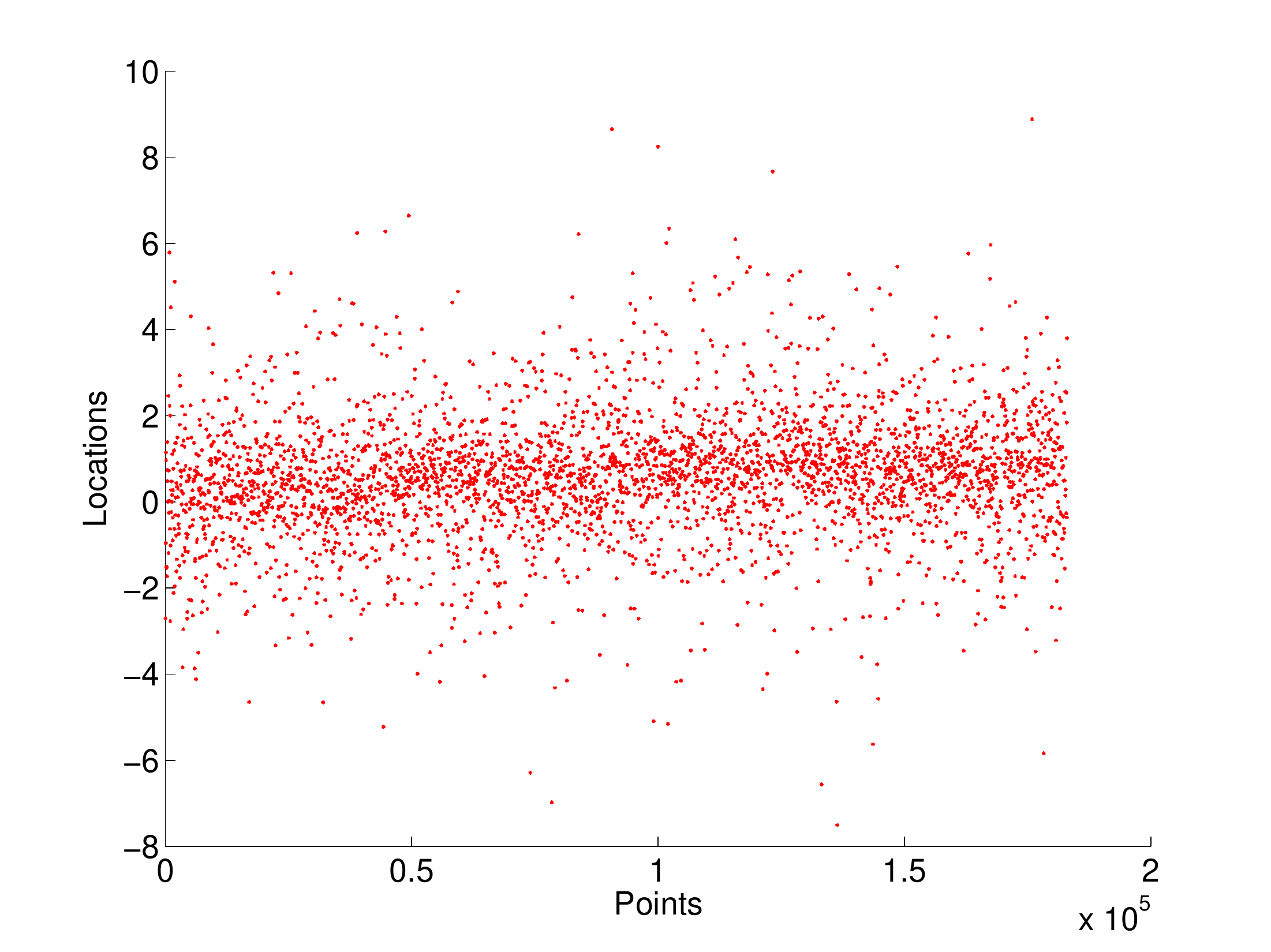}
\caption{The published noisy location data with group size $k=1$ and $\epsilon = 1$, that is, there is no grouping.  To avoid clogging, only 2\% of the published points (randomly chosen) are plotted.}
\label{fig:isotonicwithoutgroup}
\end{figure}

\begin{figure}[!h] \centering
\includegraphics[width=0.5\textwidth]{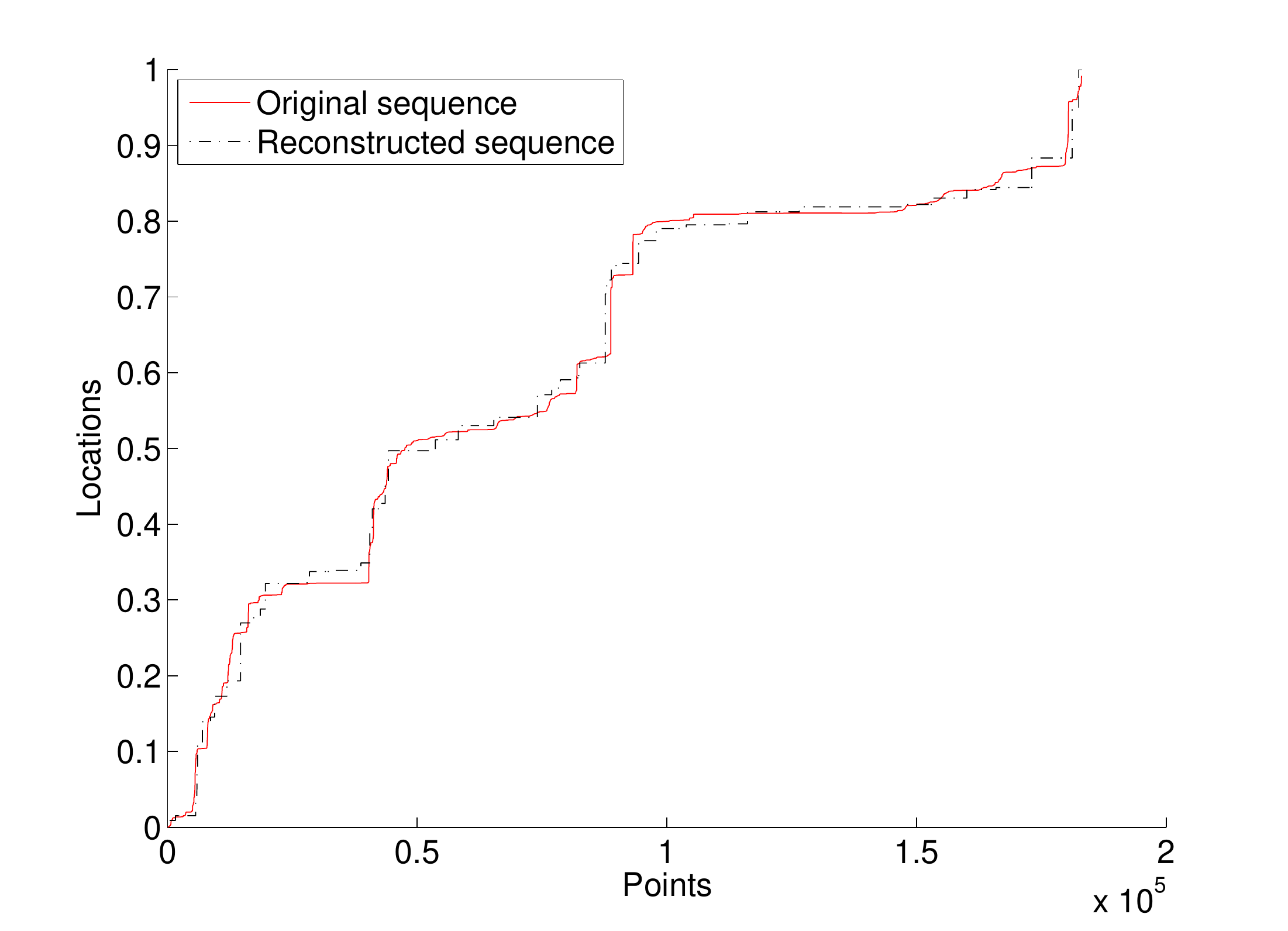}
\caption{The  reconstructed data obtained by performing isotonic regression on the published data as shown in Fig. \ref{fig:isotonicwithoutgroup},  plotted together with the original data. }
\label{fig:reconstructedwithouterror}
\end{figure}

\begin{figure}[!h] \centering
\includegraphics[width=0.5\textwidth]{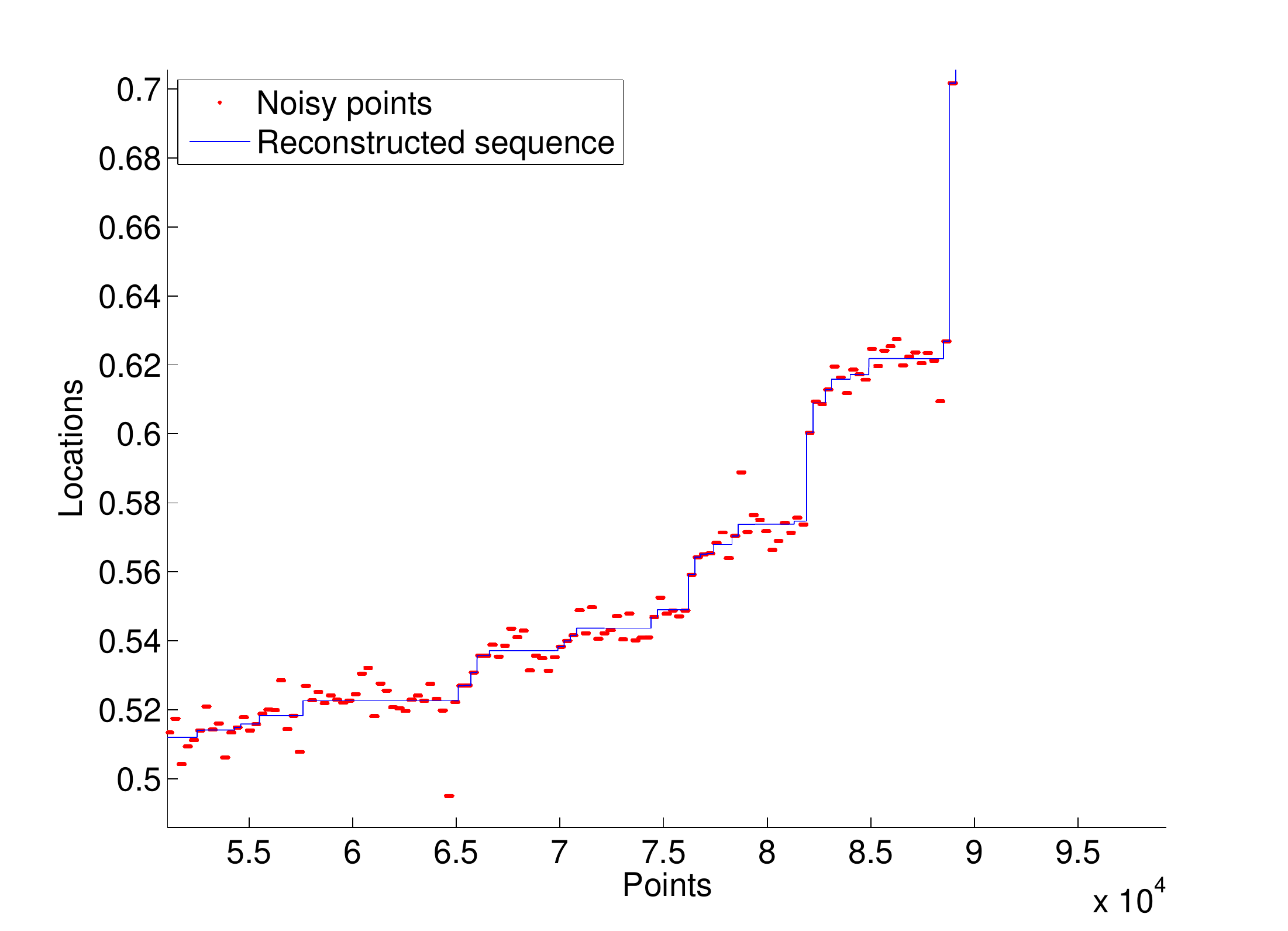}
\caption{The published noisy location data with grouping (group size $k=300$) and the corresponding reconstructed data through isotonic regression. The figure shows only the region from 0.5 to 0.7 in the unit interval.}
\label{fig:isotonicwithgroup}
\end{figure}

\begin{figure}[!h] \centering
\includegraphics[width=0.5\textwidth]{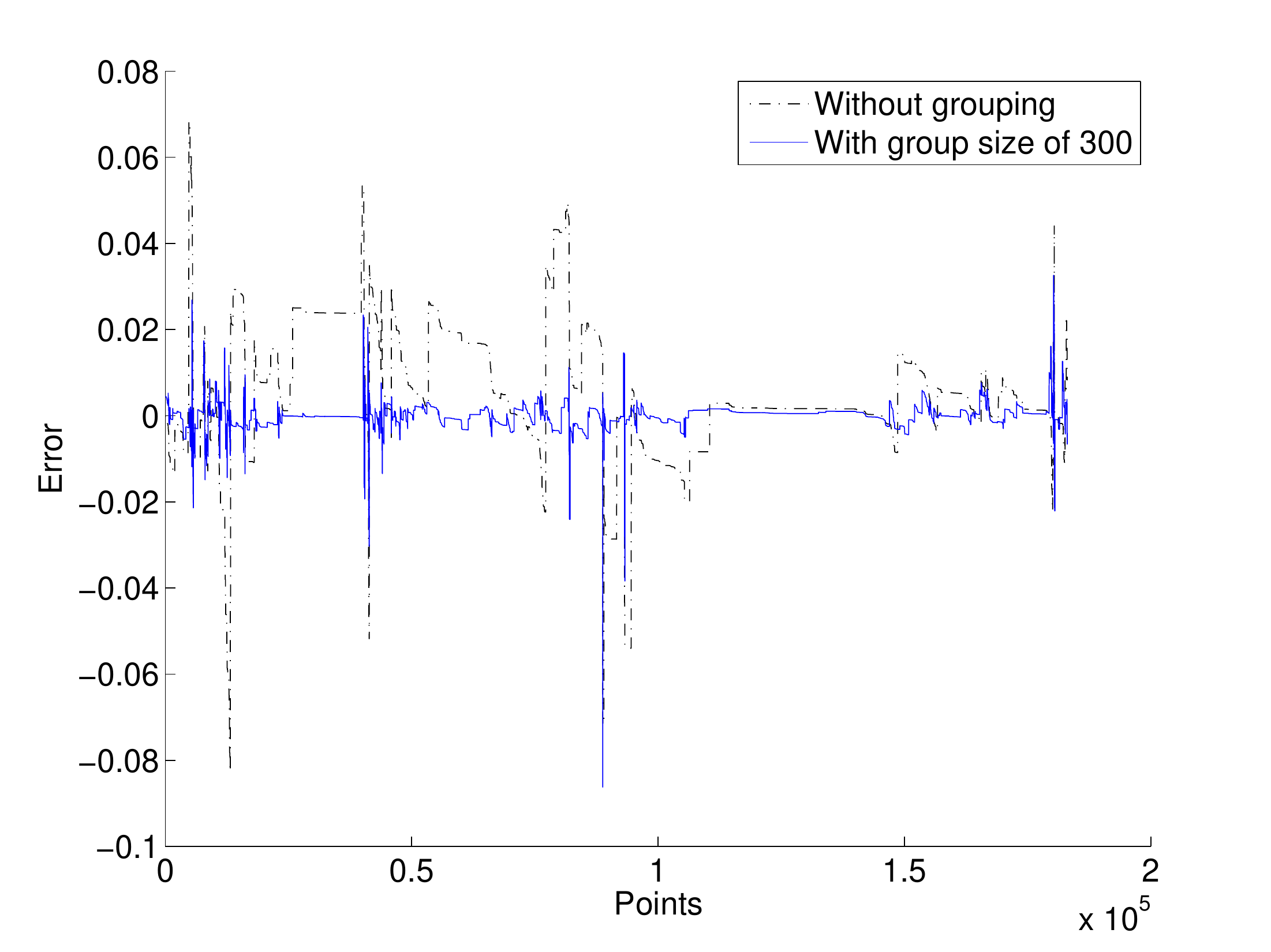}
\caption{Differences of the two reconstructed data from the original. The black dashed line is the displacement of reconstructed data without grouping (see Fig. \ref{fig:reconstructedwithouterror}), the blue solid line is the displacement of reconstructed data with group size $k=300$ (see Fig. \ref{fig:isotonicwithgroup}). }
\label{fig:compare2error}
\end{figure}

\begin{figure}[!h] \centering
\includegraphics[width=0.5\textwidth]{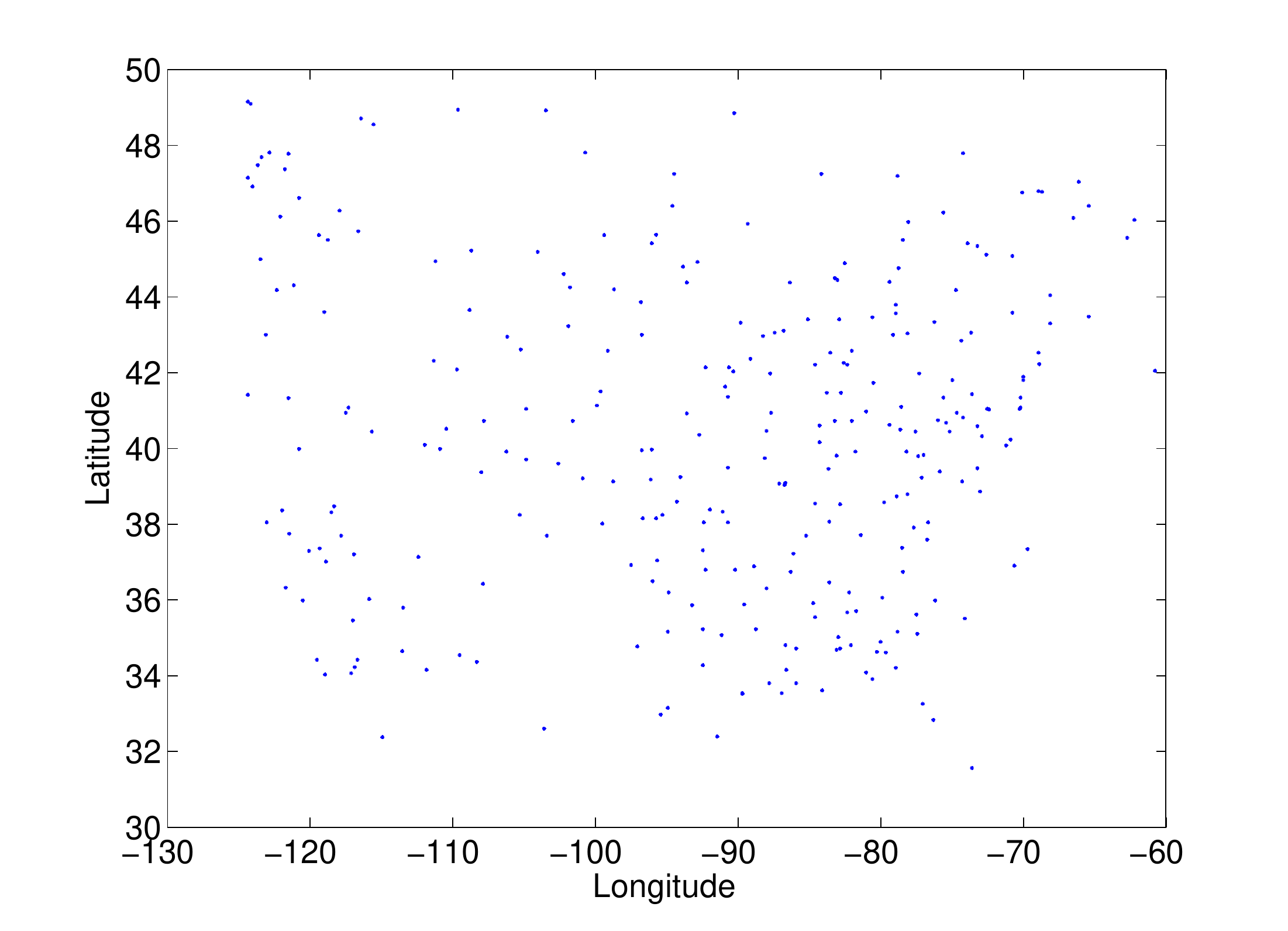}
\caption{Inverse mapped of the reconstructed data with $\epsilon = 5$. Each point in the figure represents a group of repeating points.}
\label{fig:mapusbeforepp}
\end{figure}

\begin{figure}[!h] \centering
\includegraphics[width=0.5\textwidth]{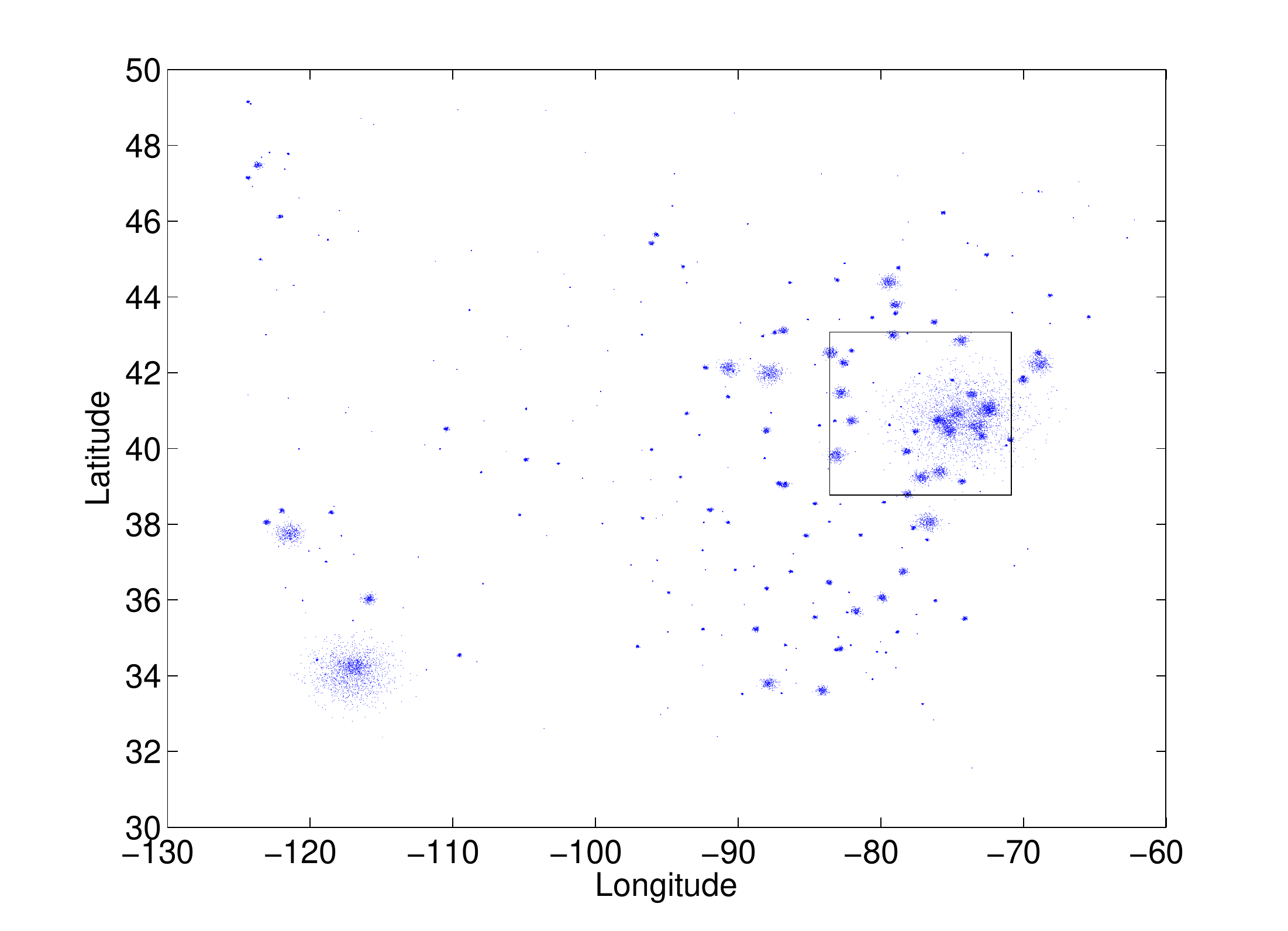}
\caption{Inverse mapped of the reconstructed data followed by post-processing. The post-processing ``diffuses'' a repeated point to its surrounding and is performed solely for visualization purpose. To avoid clogging, only 10\% of the points (randomly chosen) are plotted.}
\label{fig:reconstructed2d}
\end{figure}

\begin{figure}[!h] \centering
\includegraphics[width=0.5\textwidth]{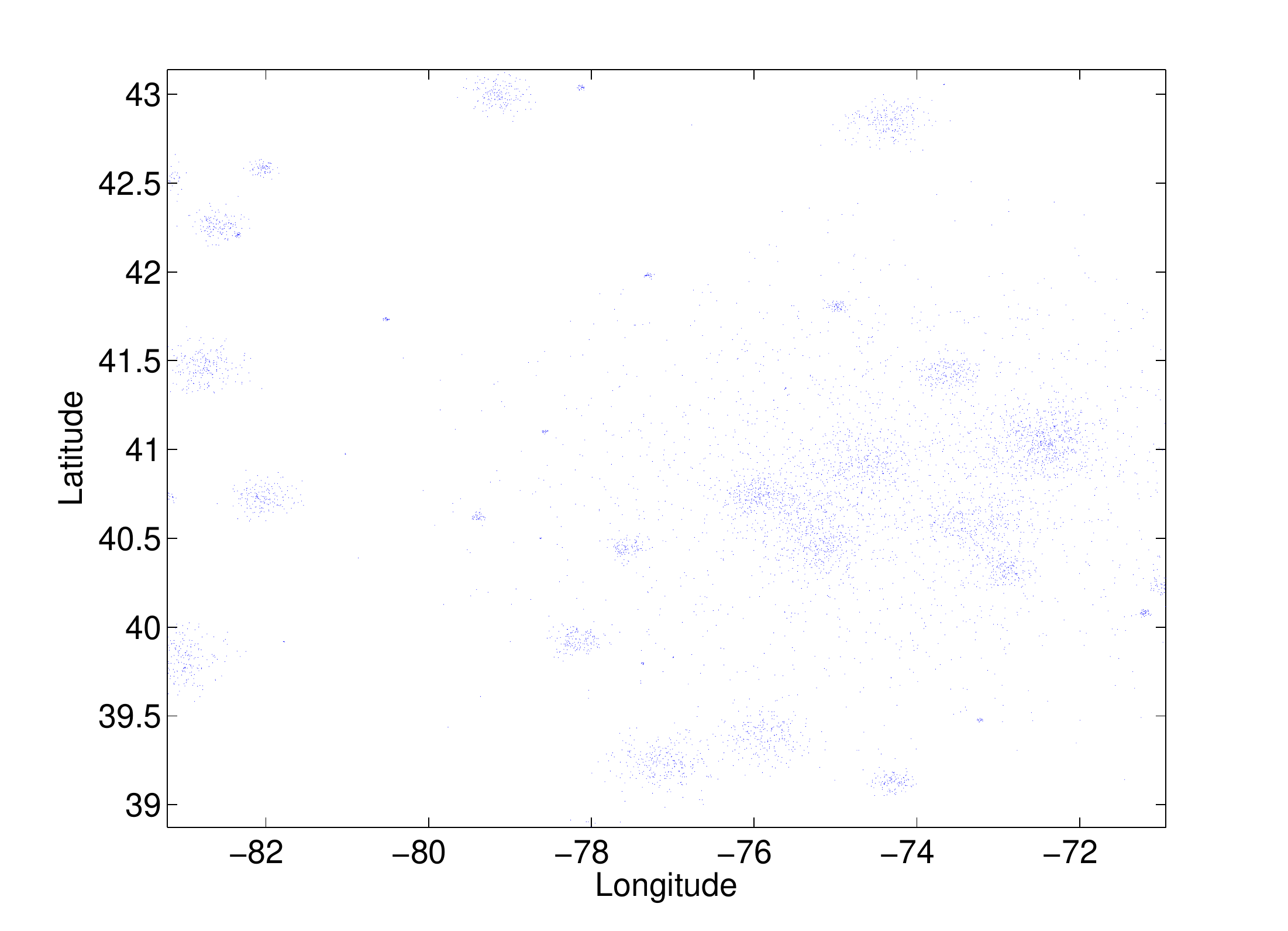}
\caption{A zoom-in of Fig. \ref{fig:reconstructed2d} to region within the indicated rectangle.}
\label{fig:eastcoast}
\end{figure}

Our choice of  the group size $k$ is determined by minimizing an error function which measures  the Earth-Mover-Distance  (EMD) of the original and reconstructed pointsets.  In one-dimension, the EMD of two equal-sized pointsets is simply the  $L_1$ distance between the two respective sorted sequences.  Although designed to minimize EMD,  the proposed mechanism achieves good accuracy w.r.t. other utilities.  Experimental studies show that the proposed mechanism achieves higher accuracy compared to the wavelet-based method for range queries\cite{xiao2010differential}, and outperforms the equi-width histogram w.r.t. the accuracy in estimation the underlying probability density function.

An advantage of the proposed mechanism is its simplicity from the publisher's viewpoint.  The publisher only has to map the points to the unit interval, sort them, add Laplace noise, and publish the results. By publishing the ``raw'' noisy data instead of the reconstructed data, users in the public are not confined to a particular inference techniques, and have the flexibility in using different variants of isotonic regression to suit their needs.

In contrast to an  equi-width histogram, the bins of an {\em equi-depth} histogram contains same number of elements, with their width varies. Intuitively,  the size of a bin is larger in location with lower ``density'' of points.  There are extensive studies on equi-depth histogram, and it is generally well accepted that an equi-depth histogram provides more useful statistical information\cite{piatetsky1984accurate} compare to equi-width histogram. However,  it is not clear how  to generate an equi-depth histogram while achieving differential privacy. Interestingly, grouping in our proposed mechanism naturally produces  equi-depth histograms: grouping of $k$ elements leads to a depth of $k$.

Our approach can be applied to obtain order-statistic, for example, median. Finding median  is challenging due to its large sensitivity.  Accurate mechanism can be derived by adding Laplace noise  proportional to the   smooth sensitivity\cite{nissim2007smooth} instead of the global sensitivity. However, computation of smooth sensitivity  takes $\Theta(n^2)$ time  where $n$ is the dataset's size. In contrast, our mechanism takes $O(n)$ time when the dataset is already sorted.  Experimental studies on datasets with $129$ elements\footnote{As it is computational intensive to compute smooth sensitivity, we are unable to repeat the experiments for significant larger $n$.}  suggest that  the proposed mechanism is less sensitive to a higher local sensitivity, or a small $\epsilon$.

The locality preserving map is a key component in our mechanism, taking the role of transforming the data points to the one-dimension space.   Although there are fundamental limits on locality preserving mapping, fortunately, our problem only requires preservation in the ``easier'' direction, i.e., any pair of neighbors in the one-dimensional domain are also neighbors in the multi-dimensional domain. The classic Hilbert space-filling curve suffices to provide high accuracy. For other types of non-spatial data, our techniques can be applied as long as an appropriate locality preserving mapping is available.\\

\noindent {\em Organization:\ \ \ }  We first describe some background materials in the next section (Section \ref{sec:background}).  In Section \ref{sec:proposedapproach} we present our main ideas and mechanism, and show that the proposed mechanism achieves differential privacy in Section \ref{sec:security}. Next, in section \ref{sec:analysis}, we formulate and analyze the noise incurred by the Laplace and the generalization noise. Based on the noise model, we derive a strategy to choose the group size.  In Section \ref{sec:experiments}, we compare our mechanism with three known mechanisms: (1) equi-width histogram, (2) wavelet-based method for range queries, and (3) smooth-sensitivity based median  finding.  In Section \ref{sec:extension}, we describe a few possible extensions, in particular, a hybrid of our mechanism with equi-width histogram.  Lastly, we describe related works in Section \ref{sec:relatedwork} and conclude in Section \ref{sec:conclusion}.

\section{Background}
\label{sec:background}
\subsection{Differential Privacy and Laplace Noise}
We treat a database as a multi-set  (i.e. a set with possibly repeating elements), and define two $D_1$ and $D_2$ to be neighbor when $D_2$ can be obtained from  $D_1$ by replacing one element, i.e.   $D_1 = \{x \} \cup  D_2  \setminus   \{ y \}$ for some $x$ and $y$. Let us call the above definition of neighborhood  the {\em replacement neighborhood}.

A randomized algorithm  (also known as a {\em mechanism}) $\mathcal{A}$  achieves $\epsilon$ differential privacy if, $$Pr[\mathcal{A}(D_1)\in S] \leq exp(\epsilon) \times
Pr[\mathcal{A}(D_2)\in S]$$ for all $S \subseteq$ Range($\mathcal{A}$), where Range($\mathcal{A}$) denotes the output range of the algorithm $\mathcal{A}$, and for any pair of  neighbouring datasets $D_1$ and $D_2$.

The replacement neighborhood we adopted is similar to the notion used by Nissim et al~\cite{nissim2007smooth}. Such variant  differs from the  well-adopted notion that  treats two datasets $D_1$, $D_2$ to be neighbors iff $D_2= D_1 \cup \{x\}$ or $D_2= D_1 \setminus \{ x\}$ for some $x$.  Note that  a mechanism that achieves differential privacy under replacement neighborhood can be converted to one that achieves privacy under  the well-adopted neighborhood.

For a function $f$ : $D \rightarrow \R^k$, the {\em sensitivity}  \cite{dwork2006differential} of $f$ is defined as $$\Delta{(f)} := \mbox{max} \| f(D_1) - f(D_2) \|_1 $$ where the maximum is taken over all  pairs of neighboring $D_1$ and $D_2$. It can be shown~\cite{dwork2006calibrating}   that the mechanism $\mathcal{A}$
\[    \mathcal{A}(D) =  f(D) + (Lap(\Delta{(f)}/\epsilon))^k \]
achieves $\epsilon$-differential privacy, where $(Lap(\Delta{(f)}/\epsilon))^k$   is a vector of  $k$ independently  and randomly chosen values from the Laplace distribution with standard deviation $2\Delta{(f)}/\epsilon$.

It is meaningless if the output of a mechanism  is simply noise, even if  privacy is achieved. The accuracy of a mechanism is measured by a {\em utility } function $u(X,y)$ that measures the quality of the output $y$  given the dataset is $X$. Alternatively, the utility can be measured by an {\em error} function that measures the distant of the output  from the ideal  output.

The notion of differential privacy has a useful {\em sequential composition property} \cite{mcsherry2009}:  if mechanisms ${\cal M}_1$ and ${\cal M}_2$  achieve $\epsilon_1$ and $\epsilon_2$-differential privacy respectively, then the combined mechanism of applying  ${\cal M}_1$ follows by ${\cal M}_2$  achieves
$\epsilon_1+\epsilon_2$ differential privacy.

\subsection{Isotonic Regression}
Given  a sequence  of $n$ real numbers $a_1,\ldots, a_n$, the problem of finding the least-square fit $x_1,\ldots, x_n$ subjected to the constraints $x_i \leq x_j$ for all $i<j\leq n$ is known as the isotonic regression. Formally, we wants to find the $x_1,\ldots, x_n$ that minimizes
$$ \sum_{i=1}^n ( x_i - a_i)^2,  \mbox{subjected to } x_i \leq x_j  \mbox{ for all } 1\leq  i <j \leq n $$
The unique solution  can be efficiently found using pool-adjacent-violators algorithms in $O(n)$ time \cite{grotzinger1984projections}. When minimizing w.r.t. $\ell$-1 norm, there is also an efficient $O(n \log n)$ algorithm\cite{quentin2000}. There are many variants of isotonic regression, for example, having a smoothness component in the objective function~\cite{wang2008,meyer2008}.

\subsection{Locality Preserving Map}
A  locality preserving map $T:  \R^d \rightarrow \R$  maps  $d$-dimensional points to real numbers while preserving ``locality''.   In this paper, we seek mapping whereby two neighboring points  in the one-dimensional range are also neighboring points in the $d$-dimensional domain. Specifically,  there is some constant $A$ s.t. for any $x, y \in \R^d$,
$$
      \| x - y\|_2  \leq A \cdot ( T(x)- T(y) )^{1/d}
$$
The well-known Hilbert curve achieves $ \| x - y \|_2 \leq 3\sqrt{| T(x) - T(y) |}-2$ for two sufficiently far-aparted points $x, y$ in $\R^2$\ \ \cite{gotsman1996metric}. Niedermeier et al. \cite{niedermeier1997towards} showed that with careful construction, the bound can be improved to $\sqrt{4|T(x)- T(y)|-2}$. In our construction, for simplicity, we use Hilbert curve.

Note that it is challenging in preserving locality in the other direction, that is,  any two neighboring points in  the $d$-dimensional domain are also neighboring points in the one-dimensional range.  Fortunately, in our problem, such property is not required.

\subsection{Datasets}
We conduct experiments on two datasets: locations of Twitter users~\cite{website} (herein called the Twitter location dataset) {\ and the dataset collected by Kalu{\v{z}}a et al.~\cite{kalu2010agent} (herein called Kalu{\v{z}}a's dataset).  The Twitter location dataset contains over 1 million Twitter users' data from the period of March 2006 to March 2010, among which around 200,000 tuples are labeled with location (represented in latitude and longitude) and most of the tuples are in the North America continent, concentrating in regions around the state of New York and California. Fig. \ref{fig:NAorigindata} shows the cropped region covering most of the North America continent. The cropped region contains 183,072 tuples. The Kalu{\v{z}}a's dataset contains 164,860 tuples collected from tags that continuously records the locations information of 5 individuals.

\section{Proposed Approach}
\label{sec:proposedapproach}

Given the privacy requirement $\epsilon$ and a dataset $D$ of size $n$, the publisher carries out the following:\\
\begin{enumerate}
\item[A1.] \label{stp:step11}  Maps each point in $D$ to a real number in the unit interval $[0,1]$ using a locality preserving map $T$. Let $T(D)$ be the set of transformed points. Determine a group size $k$ based on $n$ and $\epsilon$.  For clarity in exposition, let us assume that $k$ divides $n$.
\item[A2.] \label{stp:step12} Sorts $T(D)$.  Divides  the sorted sequence into groups of $k$ consecutive elements. 
 For each group, determines its sum.  Let the sums be $S=\langle s_1, \ldots, s_{n/k} \rangle$.
\item[A3.] \label{stp:step13} Publishes   $\widetilde{S}=S+ (\mbox{Lap} (\epsilon^{-1}) )^{(n/k)}$ and the group size $k$.\\
\end{enumerate}

An user in the public may extract information from the published data as follow:\\
\begin{enumerate}
\item[B1.] \label{stp:step21} Performs isotonic regression on $k^{-1}\widetilde{S}$, and maps the data point back to their original domain. That is, computes  $\widetilde{D}=T^{-1}( \mbox{IR} (k^{-1} \widetilde{S}))$,  where $\mbox{IR}(\cdot)$ denotes isotonic regression. Let us call $\widetilde{D}$ the reconstructed data.\\
\end{enumerate}

\subsection*{Remark.}
\begin{enumerate}

\item The size of the dataset $n$ is not considered to be a secret and can be derived from the published $\widetilde{S}$.
The transformation $T$ and the lookup table for $k$  are public knowledge prior to the publishing.

\item
\label{remark:2}
When the database size $n$ is unknown to the user,  the publisher can  exploit the sequential composition property of differential privacy  and carry out the following steps:  (1)  Firstly,  publishes a noisy size $\widetilde{n}$ using a  portion of the privacy ``budget''.  (2) Next,   extracts exactly  $\widetilde{n}$ points from the dataset using  a deterministic padding algorithm as follow: if $\widetilde{n}>n$, inserts  $(\widetilde{n}-n)$ number of $0$'s  to the dataset;  if $\widetilde{n}<n$, removes $(n-\widetilde{n})$ smallest  elements.   (3) Lastly,  publishes the  padded pointset  using our proposed mechanism.

\item To relieve the public users from computing step B1, the regression can be carried out by the publisher on behalf of the users. Nevertheless, the raw data $\widetilde{S}$ should be (but not necessary) published alongside the reconstructed data.

\item The public is not confined to adopt a particular isotonic regression. After $\widetilde{S}$ is published, various inference techniques can be applied. For instance, a user may perform a  variant of  isotonic regression that optimizes  objective functions with a smoothness component\cite{wang2008,meyer2008}.

\item
The publisher's main design decisions are the choice of $T$ and the group size $k$. The choice of $T$ depends on the underlying metric of the points. For Euclidean distance in two-dimensional space, the classic Hilbert curve already attains good performance. The group size $k$ can be computed from the lookup table constructed using our proposed noise model.
\end{enumerate}

\section{Security Analysis: Sensitivity of Sorting is bounded}
\label{sec:security}
In this section, we show that the proposed mechanism (Step A1 to A3)  achieves differential privacy, and thus also the reconstructed pointset output by B1. The following theorem shows that sorting, as a function, interestingly has sensitivity 1. Note that a straightforward analysis that treats each element independently could lead to a bound of $n$, which is too large to be useful.

\begin{theorem}
\label{thm:sensofsort}
Let $S_n(D)$ be a function that on input $D$, which is  a multiset containing $n$ real numbers from the unit interval $[0,1]$, outputs the sorted sequence of elements in $D$. The sensitivity of $S_n$ w.r.t. the replacement neighborhood is 1.
\end{theorem}

\begin{proof}
Let Let $D_1$ and $D_2$ be any two neighbouring datasets. $\langle x_1,
x_2 \ldots x_i \ldots x_n \rangle$ be  $S_n(D_1)$, i.e. the sorted  sequence of $D_1$. WLOG, let us assume that  an  element $x_i$ is replaced by a larger value $A$ to give $D_2$, for some $1 \leq i \leq n-1$ and $x_i<A$. Let $j$ to be largest index s.t. $x_j < A \leq 1$.  Hence, the sorted sequence of  $D_2$ is:
$$x_1,x_2, \ldots, x_{i-1}, x_{i+1}, \ldots, x_{j}, A, x_{j+1}, \ldots, x_{n}$$
The $L_1$ difference due to the replacement is,
\begin{eqnarray*}
&&\| S_n(D_1) - S_n(D_2)\|_1    \\
&=& |x_{i+1} - x_i| + |x_{i+2} - x_{i+1}| +  \ldots    \\
&& \ \ \ \ \ \  \ \ \ \ \ \ \ \ \ \ \ \ \ \ \ + |x_{j} - x_{j-1}| + | A - x_{j}|  \\
& = &  (x_{i+1} - x_i) + (x_{i+2} - x_{i+1}) + \ldots  \\
&& \ \ \ \ \ \  \ \ \ \ \ \ \ \ \ \ \ \ \ \ \  + (x_{j} - x_{j-1}) + (A - x_{j}) \\
& = & A - x_i \leq 1\\
\end{eqnarray*}
We can easily find an instance of $D_1$ and $D_2$ where the difference $A- x_i =1$. Hence, the sensitivity is $1$. \end{proof}

In the proof,  the fact that the sequence is sorted is exploited to obtain the bound.   Since the sensitivity is 1, the mechanism $S_n(D) + Lap(1/\epsilon)^n$ enjoys $\epsilon$-differential privacy.   Also note that the value of  $n$ is fixed. Hence, in the context of data publishing, the  size of  $D$ is not a secret and is made known to the public.

Next, we show that grouping (in Step A2) has no effect on the sensitivity.
\begin{corollary}
\label{col:grouping} Consider a partition $H = \{h_1, h_2 \ldots h_m\}$ of the indices $\{ 1, 2, \ldots, n\}$.
Let $S_H(D)$ be the function that, on input $D$, which is a multiset containing $n$ real numbers from the unit interval $[0,1]$,  outputs a sequence of $m$ numbers:
$$y_i = \sum_{j \in h_i}    x_{j},$$
for $1\leq i\leq m$   where $\langle x_1, x_2, \ldots, x_n \rangle$ is the sorted sequence of $D$. The sensitivity of $S_H$ w.r.t. the replacement neigbourhood is 1.
\end{corollary}

\begin{proof}
Let us consider two neighbouring datasets $D_1$ and $D_2$, and their respective sorted sequences be
$$x_1, x_2, \ldots, x_n$$ and
$$x'_1,x'_2, \ldots, x'_{n}$$
  WLOG, let us assume that $D_2$ is obtained by replacing an element in $D_1$ with a strictly larger element. Thus, $x_i \leq x'_i$ for all $i$'s.

The $L_1$ difference due to the replacement is
\begin{eqnarray*}
& & \|S_{H}(D_1) - S_{H}(D_2) \|_1 \\
 & = & \sum_{i=1}^m \left|  \sum_{j \in h_i} x'_i  -  \sum_{j \in h_i} x_i \right| \\
 & = & \sum_{i=1}^m \left(\sum_{j \in h_i} {x}'_i -  \sum_{j \in h_i} x_i \right) \\
& = &  \sum_{j=1}^n  \left( x'_j - x_j \right)  \leq 1\\
\end{eqnarray*}
\end{proof}

Note that the grouping in step A2  is a special partition with equal-sized $h_i$'s . Hence, Corollary \ref{col:grouping} gives a more general result where $H$ can be any partition.  From Corollary \ref{col:grouping}, the proposed mechanism that publishes  $\widetilde{S}$  achieves  $\epsilon$-differential privacy.

\section{Analysis and Parameter Determination}
\label{sec:analysis}
The main goal of this section is to analyze the effect of the privacy requirement $\epsilon$, dataset size $n$ and the group size $k$ on inducing the error in the reconstructed data, which in turn provides a strategy in choosing the parameter $k$ from the given $n$ and $\epsilon$.

Intuitively, in the absent of ``generalization noise'', when $n$ is larger, there are more constraints 
in the isotonic regression, leading to a more accurate reconstruction.
Grouping affects the accuracy in two opposing ways. It reduces the number of constraints for regression, and introduces  \emph{generalization error}.  On the other hand, the Laplace noise is essentially reduced by a factor of $k$.
By taking into account of the above factor, and the effect of generalization noise, we can determine the optimal $k$.

\subsection{Error function and Utility}
We use an error function related to the Earth-Mover-Distance (EMD)~\cite{rubner1997earth} to quality the utility of the published data.   The EMD between two pointsets of equal size is defined to be the minimum cost of bipartite matching between the two sets, where the cost of an edge linking  two points is the cost of moving one point to the other.      Hence, EMD can be viewed as  the  minimum cost of transforming one pointset to  the other.   Different variants of EMD differ on how the cost is defined.  In this paper, we adopt the typical definition that defines the cost as the Euclidean distant between the two points.

In one-dimensional space, the EMD between two sets $D$ and $\widetilde{D}$ is simply  the $L_1$ norm of the differences between the two  respective sorted sequences, i.e. $\| S_n(D) - S_n(\widetilde{D}) \|_1$, which  can be efficiently computed.  In other words,
\begin{eqnarray}
\label{eqn:pointdistance}
 \mbox{EMD} (D, \widetilde{D})= \sum_{i= 1}^{n} |p_i - \widetilde{p}_i|
\end{eqnarray}
where $p_i$'s and $\widetilde{p}_i$'s are the sorted sequence of $D$ and $\widetilde{D}$ respectively.

Given a pointset $D$ and the published pointset $\widetilde{D}$ of a mechanism ${\cal M}$ where $|D|=|\widetilde{D}|=n$, let us define the {\em normalized error} as  $\frac{1}{n}\mbox{EMD} (D, \widetilde{D})$ and denote  $\Err_{\cal M, D}$ the expected  normalized error,
\begin{eqnarray}
{{\Err}}_{{\cal M},D}
 = Exp \left[ \  \frac{1}{n}  \  \mbox{EMD}( D, \widetilde{ D} )  \  \right]
 \end{eqnarray}  where the expectation is taken over the randomness in the mechanism.

Although EMD can be computed efficiently for one-dimensional pointsets,  the best known algorithm that computes  EMD in higher dimension  has cubic running time \cite{lawler2001combinatorial}. Jang et al.~\cite{jang2011linear} proposed a fast approximation that employs a space-filling curve. Similarly, for higher dimensional space,  we approximate the  EMD  by first  map each point to a real number in $[0,1]$ through a space-filling curve, and then compute the EMD in the one-dimensional space in $O(n \log n )$ time.

\subsection{Error incurred from Laplace Noise}
Let us first omit the effect of grouping and consider cases where $k=1$.  We conduct experimental studies on four types of pointsets with varying size $n$:  (1) Multisets containing elements with the same value 0.5 (herein called ``repeating single-value dataset), (2)  sets  containing equally-spaced numbers $(i/ (n-1))$ for $i=0, \ldots, n-1$ (herein call equally-spaced dataset), (3)  sets containing  $n$ randomly chosen elements  from the Twitter location data~\cite{website}, and (4) sets containing $n$ randomly chosen elements from the  Kalu{\v{z}}a's data~\cite{kalu2010agent} .

Fig. \ref{fig:errorbar} shows the expected normalized error.  Each value on the graph is the average over 500 sample runs.
Not  surprisingly, the expected error reduces when  the number of points increases. Fig. \ref{fig:epsilongeffect} shows the expected normalized error for dataset on equally-spaced points for different $\epsilon$.  The results agree with the intuition that when $\epsilon$ is increased by a factor of $c$, the error would approximately decrease by factor of $c$ as shown in Fig. \ref{fig:Ratio}.

\subsection{Effect of Grouping on Laplace Noise}

\begin{figure}[!h] \centering
\includegraphics[width=0.53\textwidth]{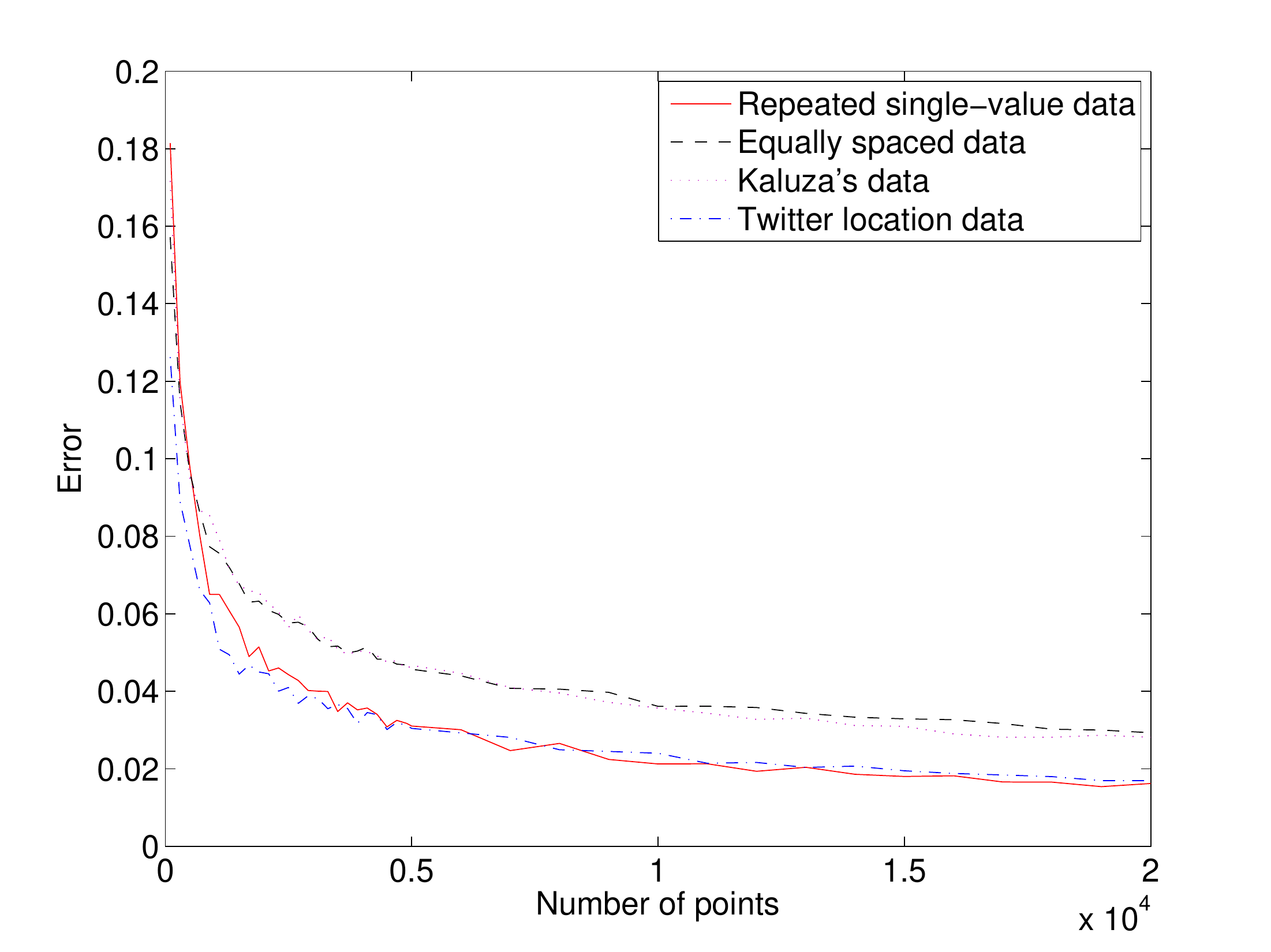}
\caption{The expected normalized error without grouping versus the size of the dataset. The red solid line is for repeating single-value dataset, the black dashed line is for equally-spaced numbers, the purple dotted line is for the  Kalu{\v{z}}a's dataset and the blue dash-dot line is for the Twitter location dataset.}
\label{fig:errorbar}
\end{figure}

\begin{figure}[!h] \centering
\includegraphics[width=0.53\textwidth]{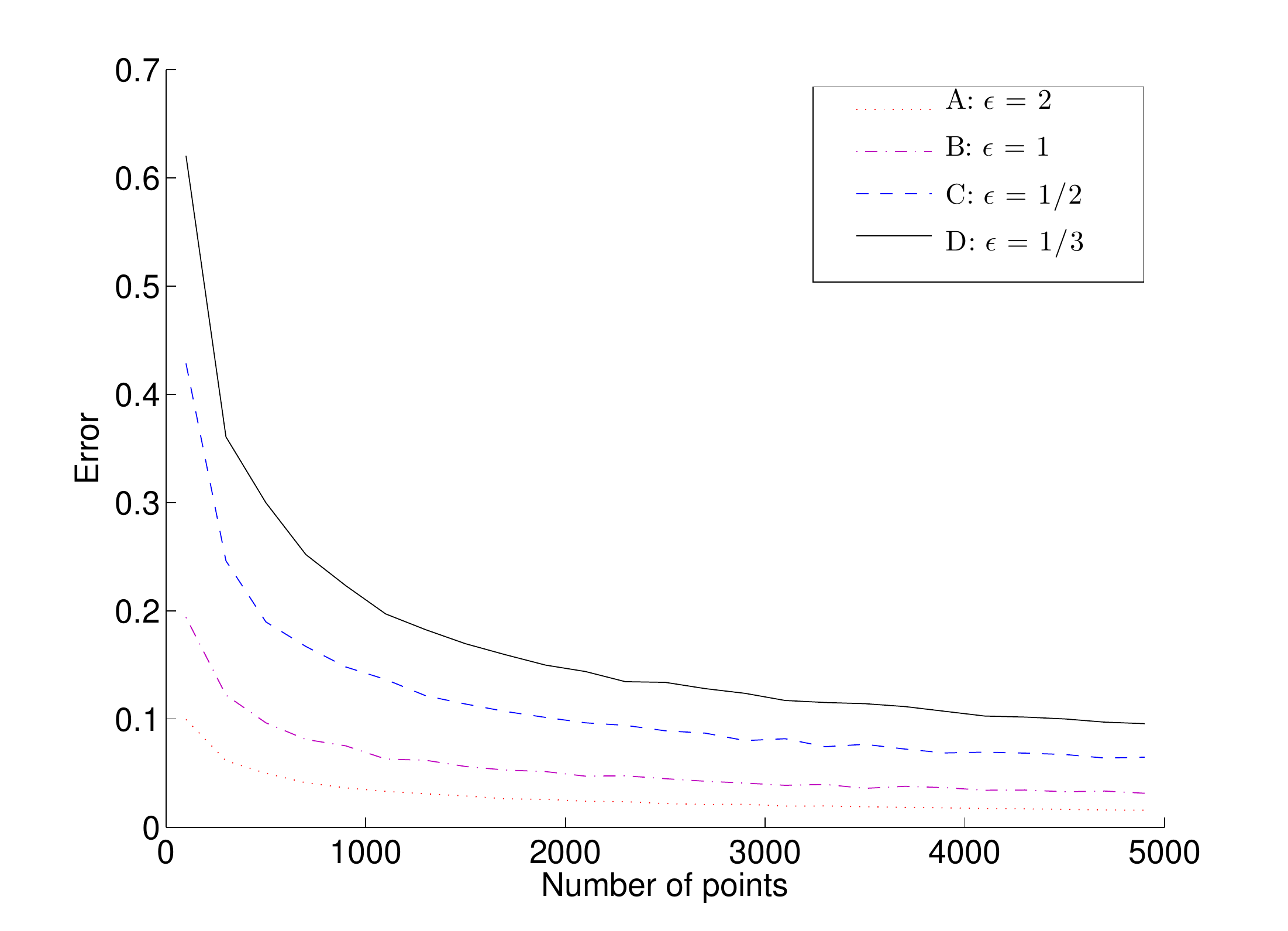}
\caption{The expected normalized error without grouping versus the size of dataset for different the security parameter $\epsilon$.}
\label{fig:epsilongeffect}
\end{figure}

\begin{figure}[!h] \centering
\includegraphics[width=0.53\textwidth]{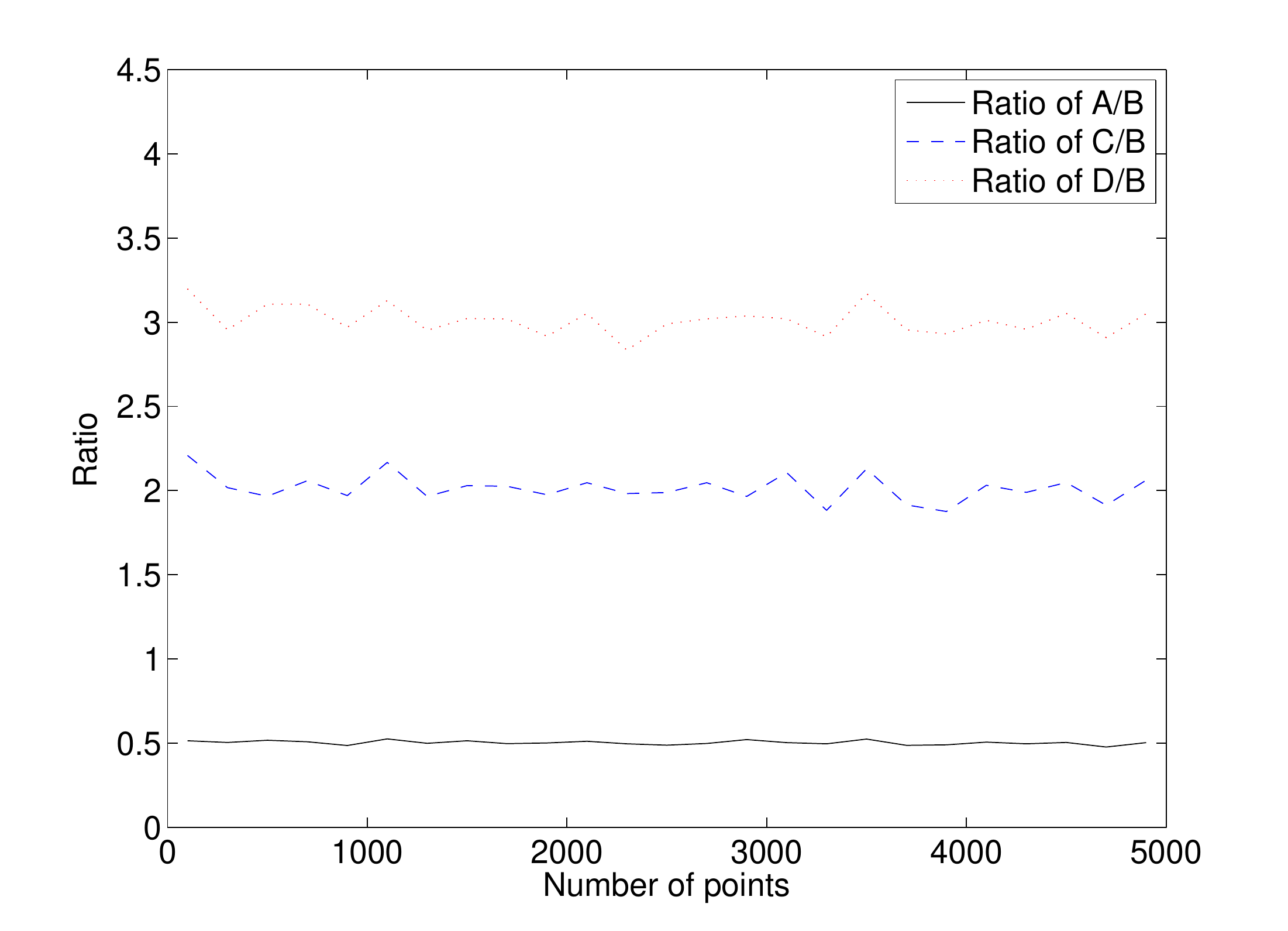}
\caption{The ratio of the expected normalized error with different $\epsilon$ against the expected normalized error with $\epsilon = 1$.}
\label{fig:Ratio}
\end{figure}

Now, we consider cases where $k>1$.  Grouping reduces the number of constrains by a factor of $k$. As suggested by Fig. \ref{fig:errorbar}, when the number of datapoint decreases,  error increases.

On the other hand, recall that the regression is performed on the published values divided by $k$ (see the role of $k$ in Step B1). This essentially reduces the level of Laplace noise by a factor of  $k$.   Hence, the accuracy attained by grouping $k$  elements is ``equivalent'' to the accuracy attained without grouping but with the privacy parameter $\epsilon$ increased by a factor of $k$.

 From Fig. \ref{fig:errorbar}, we can predict the effects of grouping on the repeating single-value dataset. For instance, if $n=10,000$ and $\epsilon=1$ and $k=5$, without grouping, the reconstructed points are expected to have a $0.02$ error; whereas with grouping of size $5$, the expected error is $0.05/5 = 0.01$.  Fig. \ref{fig:predictednoise} shows the predicted errors under different $k$'s, for $n=10,000$ and $\epsilon=1$ of different datasets.

\begin{figure}[!h] \centering
\includegraphics[width=0.53\textwidth]{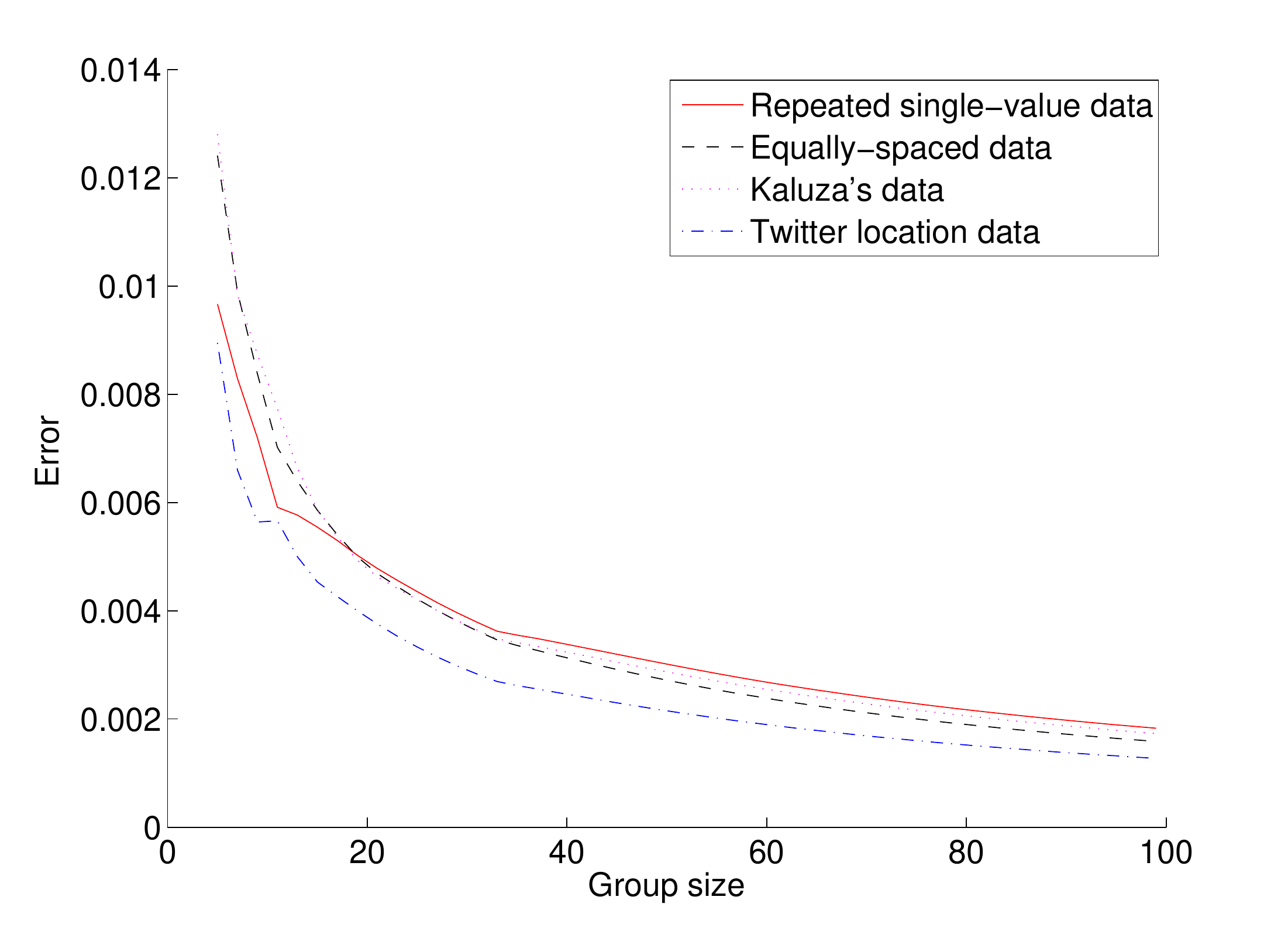}
\caption{Predicted error versus different group size without generalization noise for different datasets of size $n=10,000$ and $\epsilon=1$. The red solid line is the upper bound, the black dashed line, purple dotted line and blue dash-dot line are for equally-spaced numbers dataset, Kalu{\v{z}}a's dataset and Twitter location dataset respectively.}
\label{fig:predictednoise}
\end{figure}

\subsection {Effects of Grouping on Generalization Noise}

 The negative effect of grouping is the {\em generalization noise}, as  all elements in a group is represented by their mean.  Before giving formal description of generalization noise, let us introduce some notations.

Given a  sequence $D=\langle x_1, \ldots, x_n \rangle$ of $n$ numbers, and a parameter $k$, where $k$ divides $n$,
let us call the following function {\em downsampling}:
$$  \down_k ( D) = \langle s_1, \ldots, s_{(n/k)} \rangle$$
where each $s_i$ is the average of $x_{k(i-1)+1}, \ldots, x_{ik}$.  Given a sequence
$D'=\langle s'_1, \ldots, s'_m\rangle$ and $k$, let us call the following function  {\em upsampling},
$$ \up_k (D') = \langle x'_1, \ldots, x_{mk} \rangle$$
where $x'_i=   s'_{ \lfloor (i-1)/k \rfloor +1}$ for each $i$.

The {\em normalized generalization error} is defined as,
$$
\Gen_{D,k} = \frac{1}{n}  \| D -   \up_k ( \down_k ( D))  \|_1
 $$

It is easy to see that, for any $k$ and  $D$, the normalized generalization error is at most  $k /(2n)$.
Fig. \ref{fig:generror} shows the generalization error of different group size a dataset containing $10,000$ equally-spaced values, a dataset containing $10,000$ numbers randomly drawn from the transformed Kalu{\v{z}}a's dataset and a dataset of $10,000$ numbers randomly drawn from the transformed Twitter location data.  They agrees with our upper bound on the generalization error.

Furthermore, the worst case occurred when the values in the groups divided equally between two values, for example, half of them have value $0$, and half of them have value $1$. This is very unlikely. Intuitively, even if the elements in the groups only have two distinct value $a$ and $b$, the number of elements having value $a$ may vary. Hence, one would expect the average generalization error to be $k /(4n)$. Fig. \ref{fig:generror} shows that such approximation is very accurate and consistent for various datasets.

\begin{figure}[!h] \centering
\includegraphics[width=0.53\textwidth]{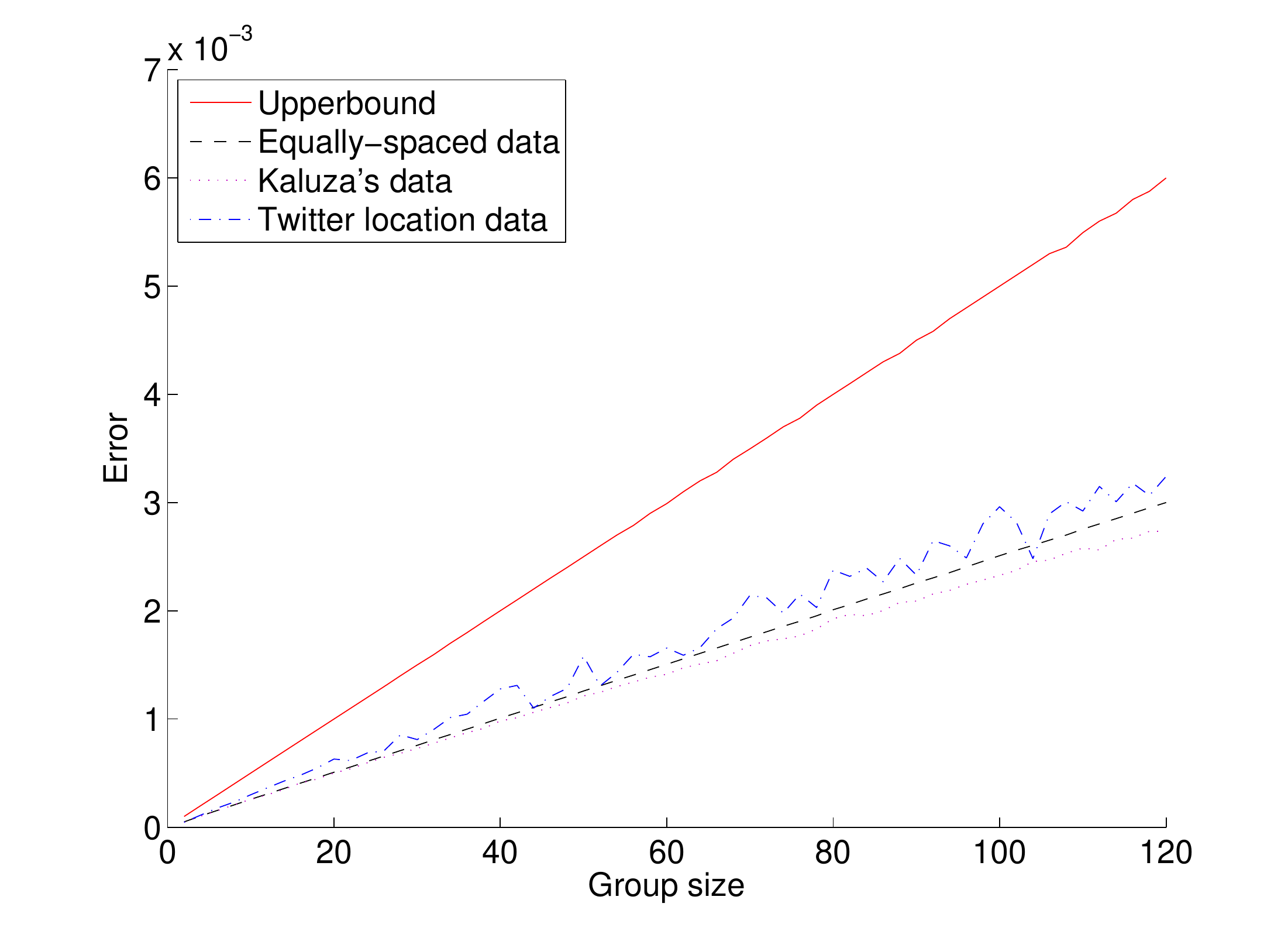}
\caption{Generalization error of different group size for different datasets of size $n=10,000$ and $\epsilon=1$.  The red solid line is for the Kalu{\v{z}}a's dataset,  the blue dotted line is for the equally-spaced dataset and the black dashed line for the  Twitter location dataset.}
\label{fig:generror}
\end{figure}

\subsection{Combined effects of grouping}

Now, let us combine the effects of both grouping and Laplace noises on the normalized error $\Err_D$.
Let us consider the mechanism that, on input $D$ and the parameter $k$, outputs
$$  {\cal M}_k (D)=  \up_k( \mbox{IR} (  \down_k( S_n(D)   ) + \mbox{Lap} (1)^{n/k} ))
 $$
This mechanism is essentially similar to our proposed method, but with the difference on how $k$ is chosen: here, the $k$ is given as a parameter, whereas in  Step A1 of the proposed method, the $k$ is chosen from a lookup table.  Recall that the expected normalized error produced by this mechanism ${\cal M}_k$  on $D$ is denoted  as
$\Err_{{\cal M}_k, D}$. For abbreviation, we write it as
$$\Err_{k, D}
$$

Let $\widetilde{S}$ to be an instance of $  \down_k( S_n(D)   ) + \mbox{Lap} (1)^{n/k}   $, and $\widetilde{D}$ the corresponding reconstructed dataset generated by ${\cal M}_k$, i.e. $\widetilde{D} = \up_k (  \mbox{IR} ( \widetilde{S}) )$. We have,
 \begin{eqnarray}
\lefteqn{\mbox{ EMD }(D, \widetilde{D} )
=  \|  S_n(D) -    \up_k (  \mbox{IR} ( \widetilde{S}) ) \|_1 }\nonumber \\
&=& \|  S_n(D)-\up_k (\down_k(S_n(D))  \nonumber \\
&&   \ \ \ \ \ \ \ \ \ \   + \up_k(\down(S_n(D)))    -  \up_k (  \mbox{IR} ( \widetilde{S} ))  \|_1\nonumber \\
&\leq&
 n \cdot  \Gen_{D,k} + \| \up_k ( \down_k ( S_n(D))) -\up_k( \mbox{IR} ( \widetilde{S})) \|_1 \nonumber \\
&=& n \cdot \Gen_{D,k} + k  \cdot \| \down_k ( S_n(D)) - \mbox{IR} ( \widetilde{S}) \|_1 \nonumber \\
&=& n \cdot \Gen_{D,k}+ k \cdot \mbox{EMD} ( \down_k(S_n(D)),   \mbox{IR} ( \widetilde{S}))
  \end{eqnarray}
Note that the first term $n\cdot  \Gen_{D,k}$ is a constant independent of the random choices made by the mechanism.  Also note that the second term is the EMD between the down-sampled dataset and its reconstructed copy obtained using  group size~1.  By taking expectation over randomness of the  mechanism (i.e. the Laplace noise  $\mbox{Lap} (1)^{n/k}$),  we have
\begin{eqnarray}
\label{eq:main_inequality}
\Err_{k, D} \leq \Gen_{k, D} + \Err_{ 1,\down_k(D)}
\end{eqnarray}
In other words,  the expected  normalized error  is bounded  by the sum of normalized generalization error,  and  the normalized error incurred by the Laplace noise. Figure \ref{fig:bestgroupsize10000} shows the three values versus different group size $k$ for equally-spaced data of size 10,000. The minimum of the expected normalized error suggests the optimal group size $k$.

Fig. \ref{fig:boundonerror} illustrates the  expected errors  for different $k$ on the Twitter location data with 10,000 points.   The red dotted line is $\Err_{k,D}$ whereas the blue solid line is the sum in the right-hand-side of the inequality (\ref{eq:main_inequality}).  Note that the differences between the two graphs are small. We have conducted experiments on other datasets and observed similar small differences. Hence, we  take the sum  as an approximation to the expected normalized error,
\begin{eqnarray}
\label{eq:main_approx}
\Err_{k,D} \approx \Gen_{k,D} + \Err_{1,\down_k(D)}
\end{eqnarray}

\begin{figure}[!h] \centering
\includegraphics[width=0.53\textwidth]{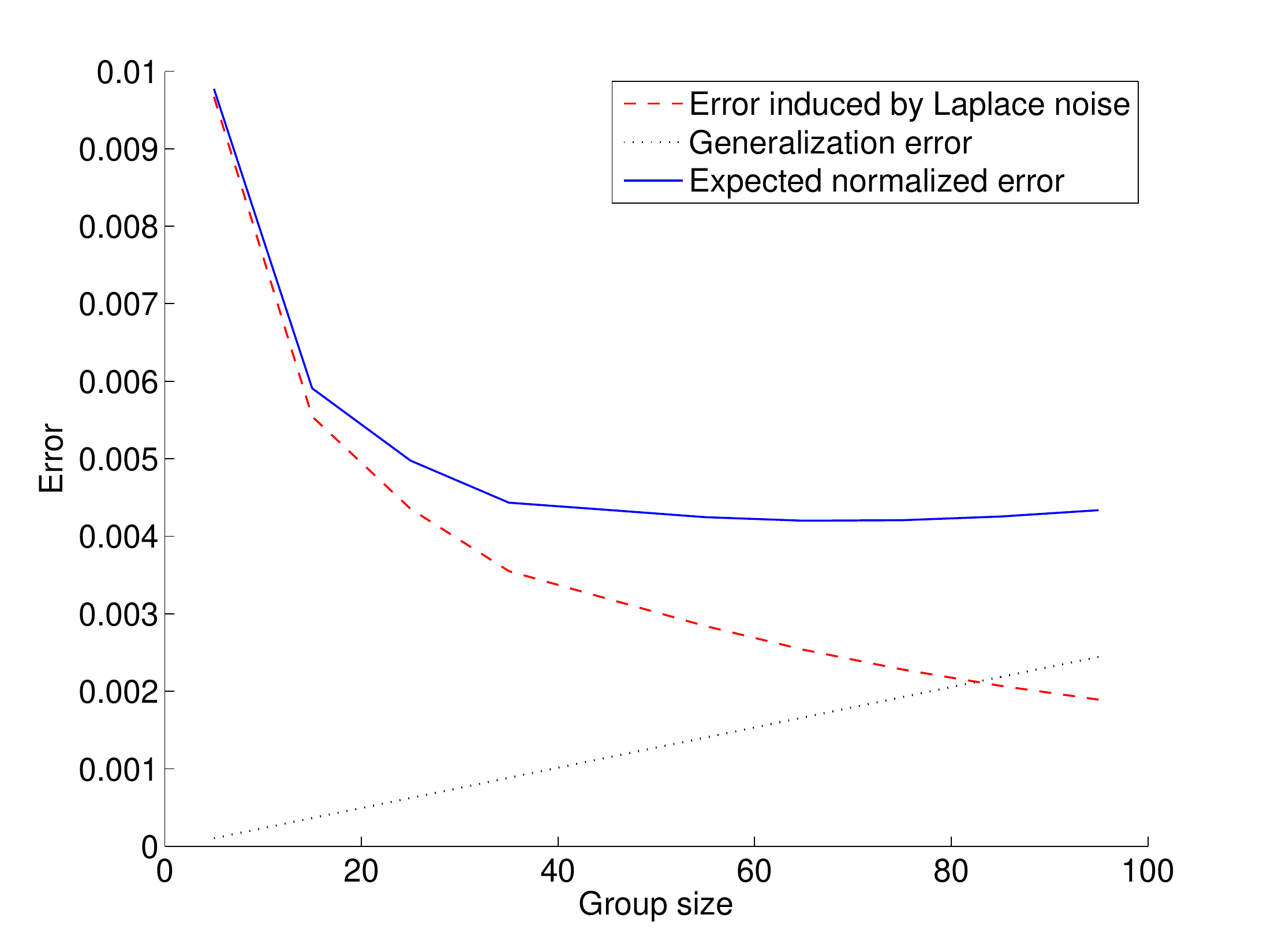}
\caption{The normalized error derived from the generalization error and perturbation error of different group size $k$ for database of size $n=10,000$ and $\epsilon=1$.}
\label{fig:bestgroupsize10000}
\end{figure}

\subsection{Determining the group size $k$}

Now, we are ready to find the optimal $k$ given $\epsilon$ and $n$.
From Fig. \ref{fig:errorbar} and Fig. \ref{fig:generror} and the approximation given in equation (\ref{eq:main_approx}), we can determine the best group size $k$ give the size of the database $n$ and the security requirement $\epsilon$. From $\epsilon$ and Fig. \ref{fig:errorbar}, we can obtain the value $\Err_{1,\down_k(D)}/\epsilon$ for different $k$. From the database's size $n$ and Fig. \ref{fig:generror}, we can approximate $\Gen_{k,D}$ for different $n$. Thus, we can approximate the normalized error $\Err_{k,D}$ with equation (\ref{eq:main_approx}) as illustrated in Fig. \ref{fig:bestgroupsize10000}.
Using the same approach, the best group size given different $n$ and $\epsilon$ can be calculated and is presented in table \ref{tab:bestgroupsize}.

\begin{figure}[!h] \centering
\includegraphics[width=0.53\textwidth]{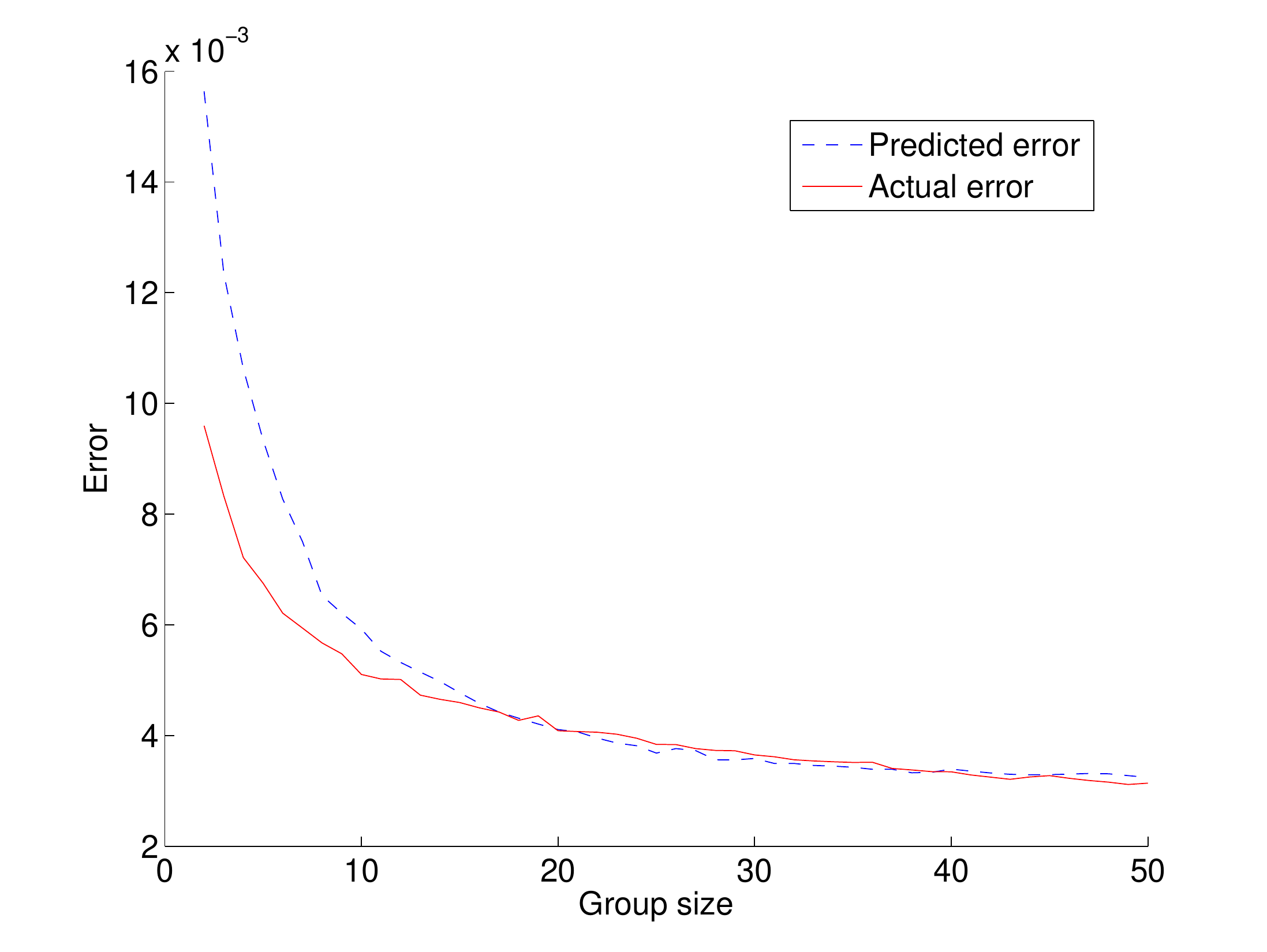}
\caption{The predicted error for 10,000 points with $\epsilon = 1$ and the actual error of dataset contains 10,000 points randomly selected from Twitter location dataset.}
\label{fig:boundonerror}
\end{figure}

\begin{table}[!h]
\centering
\caption{The best group size $k$ given $n$ and $\epsilon$}
\label{tab:bestgroupsize}
\begin{tabular}{|l|c|c|c|c|}
  \hline
   & $\epsilon = 0.5$ & $\epsilon = 1$ & $\epsilon = 2$  & $\epsilon = 3$ \\ \hline
 n= 2,000 & 44 &29 & 20 & 12 \\
 n= 5,000 &59  &37& 27  & 18 \\
 n= 10,000 &79&51 & 36  & 27 \\
 n= 20,000 & 121&83 & 61  & 41 \\
 n= 100,000 &234 & 150 & 98 & 73 \\
  \hline
\end{tabular}
\end{table}

\subsection{Effect of Isotonic Regression}

Isotonic regression is not unbiased.  Reconstructed points on the left side (i.e. having value smaller than the median) tend to have negative bias,  whereas points on the right side (i.e. having value larger than the median) tend to have positive bias. The biasness usually is smaller for points nearer to the median or the two ends.

We conducted experiments on an equally-spaced data of size 1000 to measure the displacement (the difference between the reconstructed points and the original point).  Fig. \ref{fig:bias} shows the estimated distribution of the displacement of a smaller point (the 100th point) and a larger point (the 900th point) derived from 200,000 runs of the experiments. The displacement of the smaller point has a mean of $-0.0163$ and the larger point has a mean of $0.0162$, both with a variance of $0.0138$.


\begin{figure}[!h]
\includegraphics[width=0.5\textwidth]{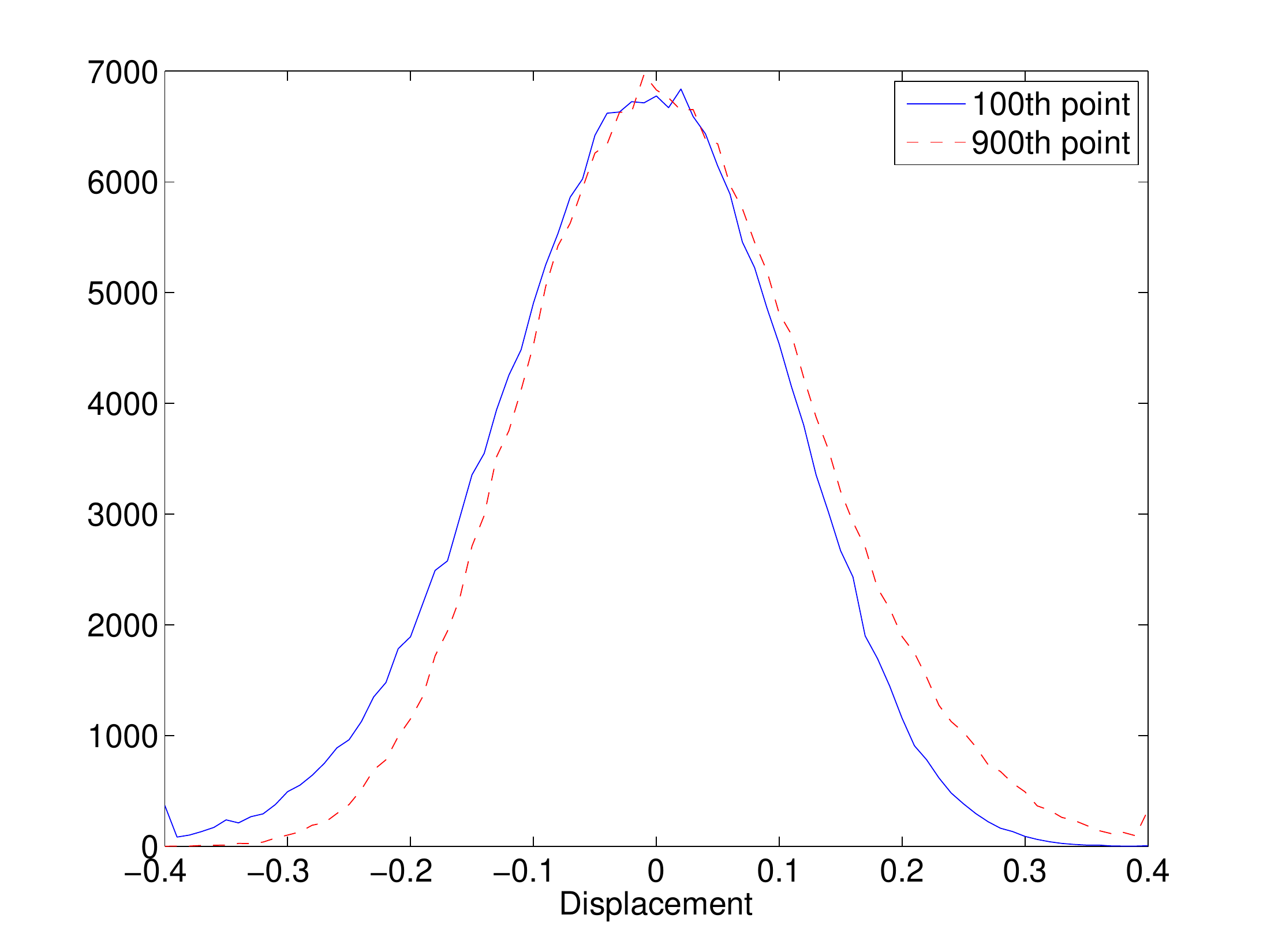}
\caption{The displacement distribution of the 100th point and the 900th point in the reconstructed sequence of our process on an equally-spaced pointset of size 1000.}
\label{fig:bias}
\end{figure}

\section{Comparisons}
\label{sec:experiments}
In this section, we compare the performance of  the proposed mechanism with three mechanisms w.r.t. different utility functions. The first mechanism outputs  equi-width histograms differential privately. We treat the generated equi-width histogram as an estimate of the underlying probability density function, and use the statistical distance between density functions as a measure of utility.  Next, we investigate the wavelet-based mechanism proposed by Xiao et al.~\cite{xiao2010differential} and measure accuracy of range queries.  Lastly, we consider the problem of estimating median, and compare with a mechanism based on smooth sensitivity proposed by Nissim et al\cite{nissim2007smooth}.  We remark  that although comparisons are based on different utility functions,  our proposed mechanism is the  same, in particular, the parameter $k$ is chosen from the same lookup table.



\subsection{Equi-width Histogram}\label{sec:density_function}
The equi-width histogram compresses of equal-sized non-overlapping bins.
With respect to the replacement neighborhood, the histogram generation has sensitivity of 2 and thus adding the Laplace noise $Lap(2/\epsilon)$ to each of the frequency counts will give a differentially private histogram.  Note that the size of the bins has to be determined prior to publishing. Without good background knowledge of the pointset, it is  difficult to determine a good choice of  bin size, as  the same bin size can lead to significantly different accuracy for different pointsets.   \\

\noindent {\em Estimating the underlying pdf:\ \ }
 A histogram of a pointset can be treated as an estimate of an underlying density function $f$ whereby the points are drawn from.  The value of  ${f}(x)$  can be estimated  by  the  ratio of the number of samples,  over the width of the bin where $x$ belongs to, multiplies by some normalizing constant. Essentially, this estimation employs step function as its kernel in estimating the density function.

In this section, we treat the estimation of the density function as the main usage of the published data. Hence, we qualify the mechanism's utility  by the distance between the two estimated density functions: one that is derived from the original dataset, and the other that is derived from the mechanism's output.


To facilitate comparison, we need an algorithm to estimate the density function from the original pointset $D$.
There are many ways to estimate the density function, and we adopt the following method: Let $B$ be the set of {\em distinct}-points in $D$, and let us consider the Voronoi diagram of $B$. The  cells in the Voronoi diagram are taken to be the bins of a histogram, from which an estimate of the density function is obtained. Note that the bins generated have variable sizes, and thus the above process can be treated as a form of  ``variable-bandwidth'' kernel density estimation\cite{terell1992}, where the kernels are step functions with different shapes.

Similarly we  need to estimate the  density function when given the output from our mechanism.
Since isotonic regression is performed on the space-filling-curve, we adopt a variant of the above estimation with the Voronoi diagram computed in the transformed one-dimensional space.  In other words,
given the reconstructed $\widetilde{D}$ of multidimensional points, let ${B}$ be the set of distinct-values in $T(D)$ where $T$ is the locality preserving map.  Next, determine the Voronoi diagram of $B$, which comprises of a sequence of intervals.  Such sequence of intervals form the bins of a histogram, from which an estimate of the  density function is obtained.\\

\noindent {\em Experimental results:\ \ }
Fig. \ref{fig:densmaporignal}, \ref{fig:densmaphist} and \ref{fig:densmapours} show the estimated density function from the Twitter's location dataset,  the density functions reconstructed by noisy equi-width histogram and the dataset by our mechanism. For comparisons, 1\% of the original points are plotted on top of the two reconstructed density functions.
For the original dataset and the reconstructed dataset by our mechanism, we quantize the location domain into $1024 \times 1024$ units, and compute the estimated density function using the aforementioned method. For the equi-width histogram, each bin is of $25 \times 25$ units, there there are in total $1681$ bins.
Fig. \ref{fig:densmaphistzoom} and \ref{fig:densmapourszoom} show the details of the two reconstructed density functions at the region $[420, 560] \times [720, 840]$.
Observe that in the density function produced by our mechanism has ``variable-sized'' cells and thus is able to adaptively capture the fine details with small amount of noise.

\begin{figure}[!h]
\includegraphics[width=0.5\textwidth]{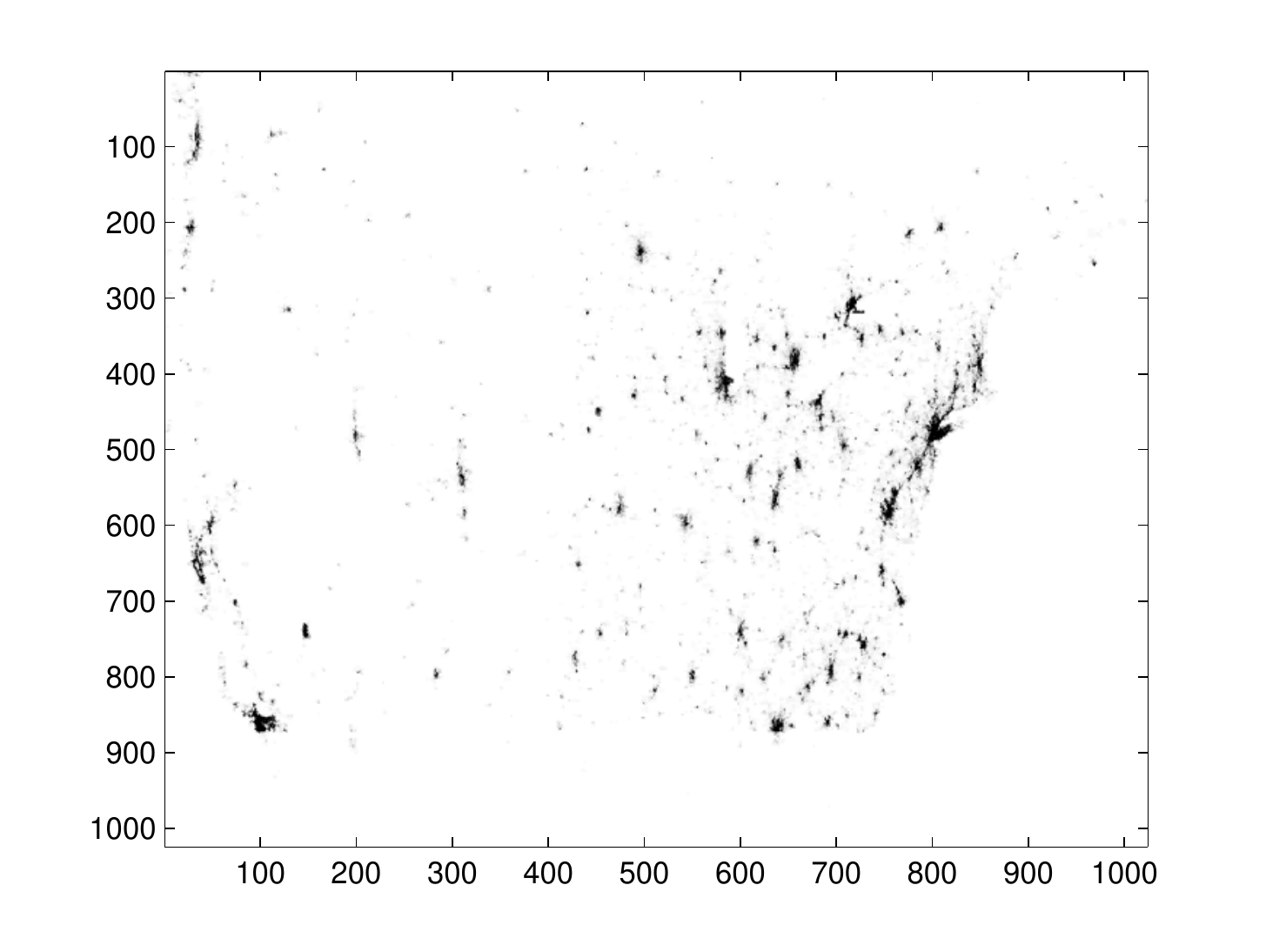}
\caption{Density function estimated from the original dataset, with darker area representing larger value.} \label{fig:densmaporignal}
\end{figure}

\begin{figure}[!h]
\includegraphics[width=0.5\textwidth]{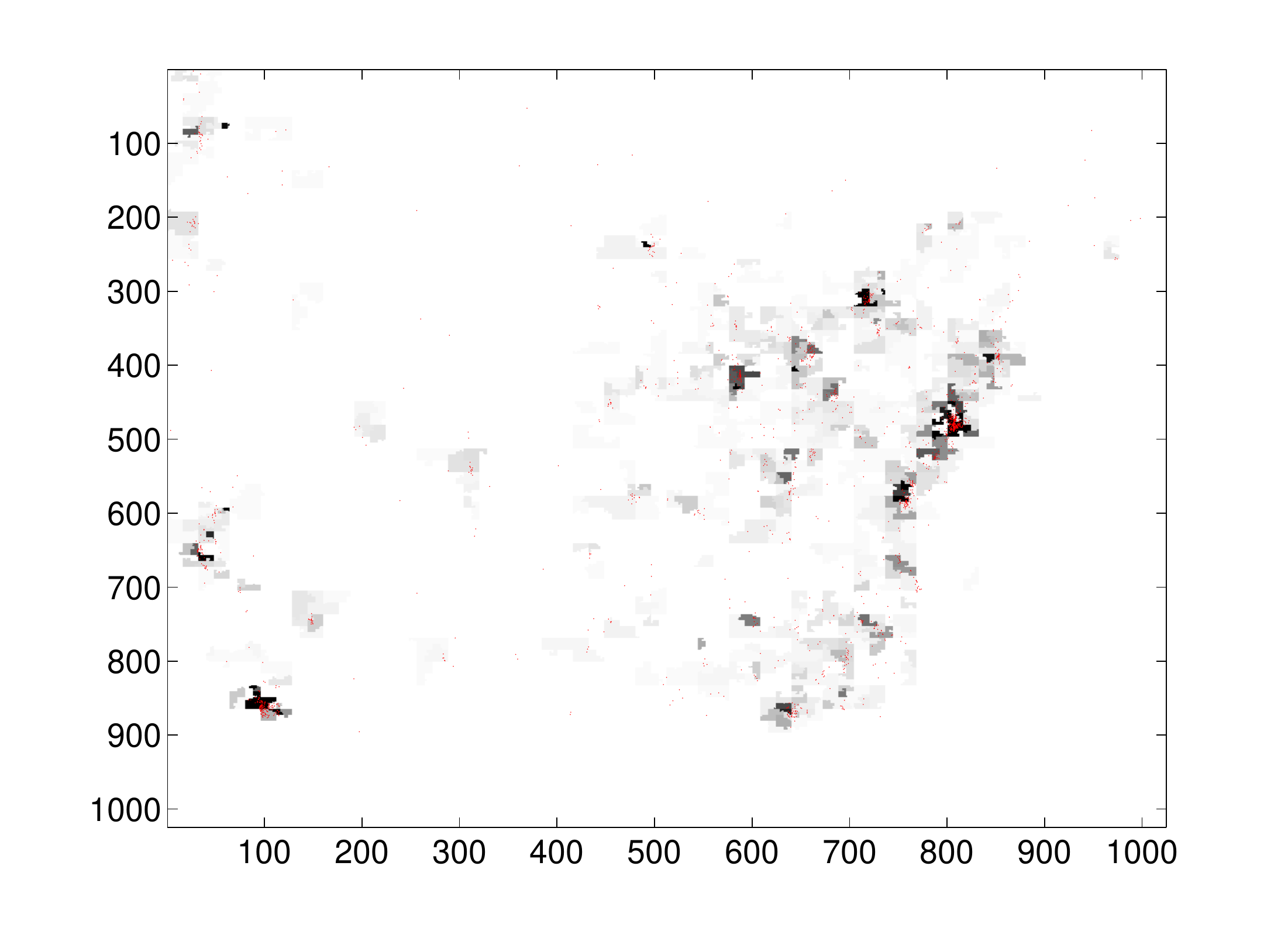}
\caption{Density function of dataset estimated using our mechanism with $\epsilon = 3$, with 1\% of the original points (randomly selected) plotted on top.} \label{fig:densmaphist}
\end{figure}

\begin{figure}[!h]
\includegraphics[width=0.5\textwidth]{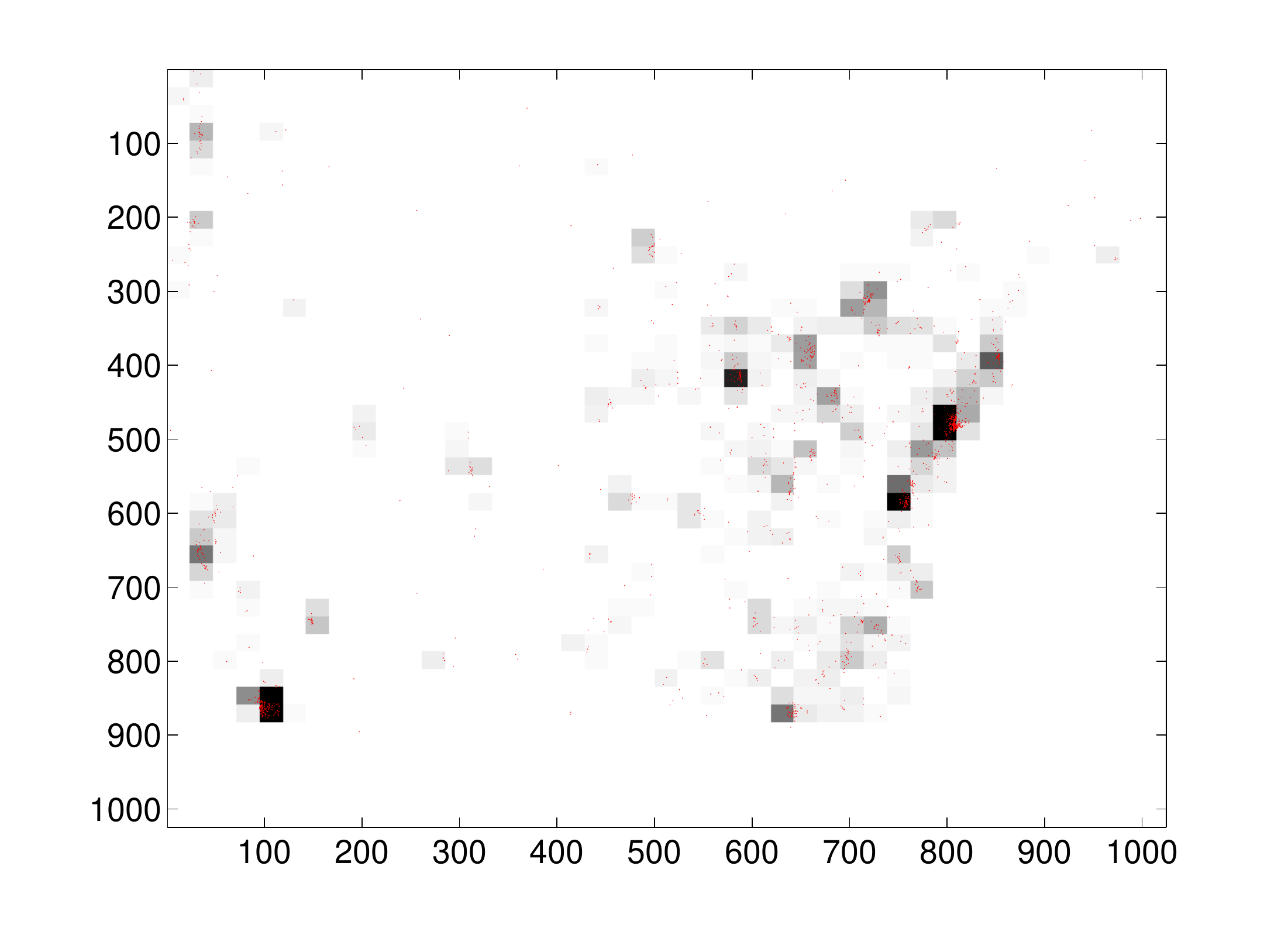}
\caption{Density function estimated using equi-width histogram mechanism with  $\epsilon = 3$, with 1\% of the original points (randomly selected) plotted on top.} \label{fig:densmapours}
\end{figure}

\begin{figure}[!h]
\includegraphics[width=0.5\textwidth]{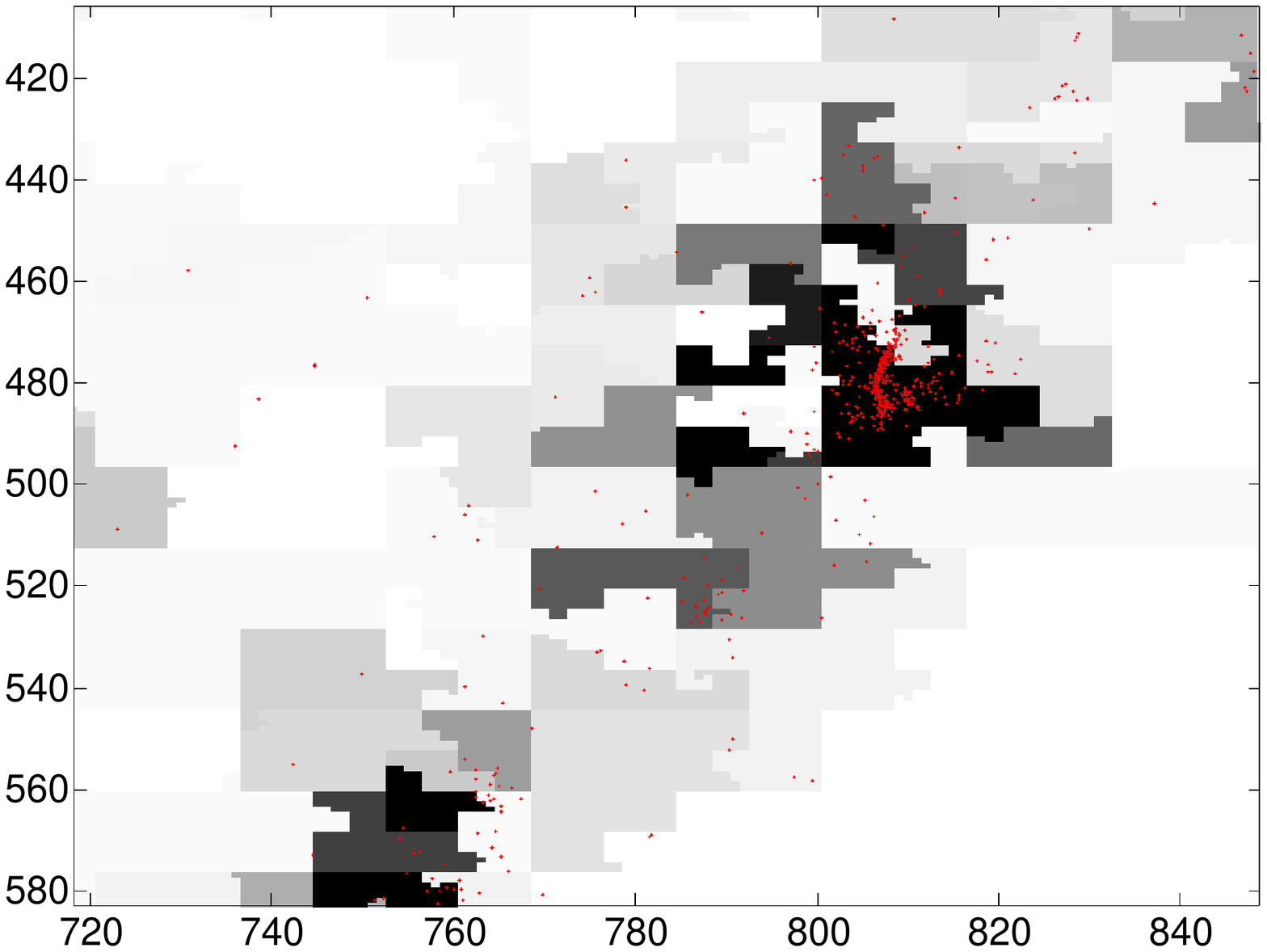}
\caption{A zoom-in view of Fig. \ref{fig:densmaphist}.} \label{fig:densmaphistzoom}
\end{figure}

\begin{figure}[!h]
\includegraphics[width=0.5\textwidth]{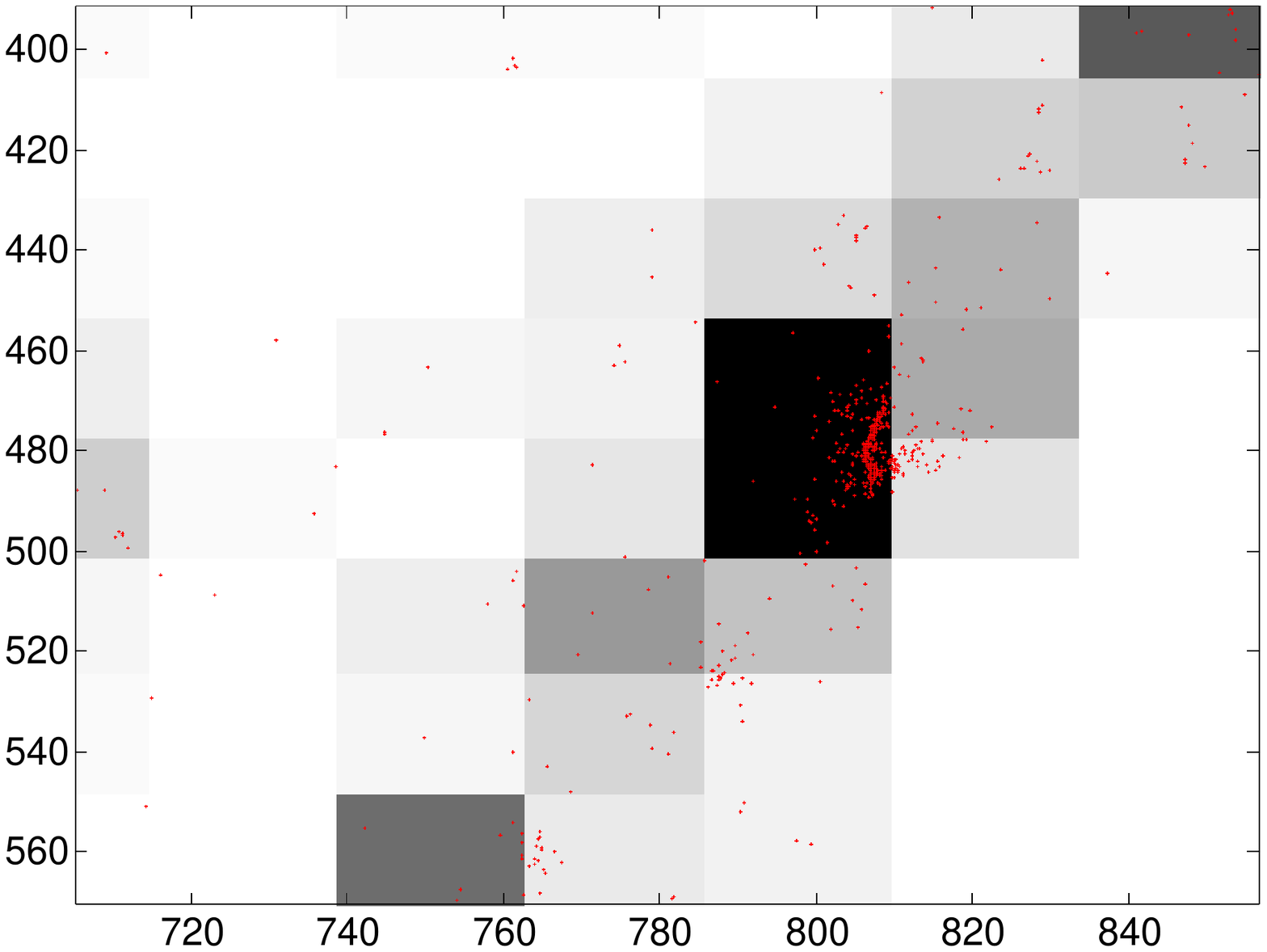}
\caption{A zoom-in view of Fig. \ref{fig:densmapours}.} \label{fig:densmapourszoom}
\end{figure}

The statistical difference, measured with $\ell_1$-norm and $\ell_2$-norm, between the two estimated density functions derived from the original and the mechanism's output are shown in Table \ref{tab:bestgroupsize}.
We remark that it is not easy to determine the optimal  bin size for the equi-width histogram prior to publishing.  Figure \ref{fig:histodistribution} shows that the optimal bin size differs significantly for three different datasets.

\begin{table}
\centering
\caption{The statistic difference between the estimated density function.}
\label{tab:bestgroupsize}
\begin{tabular}{|c|c|c|}
  \hline
   & equi-width histogram method & proposed method \\
  $\ell_1$-norm & 1.38 &  1.19  \\
  $\ell_2$-norm & 0.25 &  0.20  \\
  \hline
\end{tabular}
\end{table}

\begin{figure}[!h]
\includegraphics[width=0.5\textwidth]{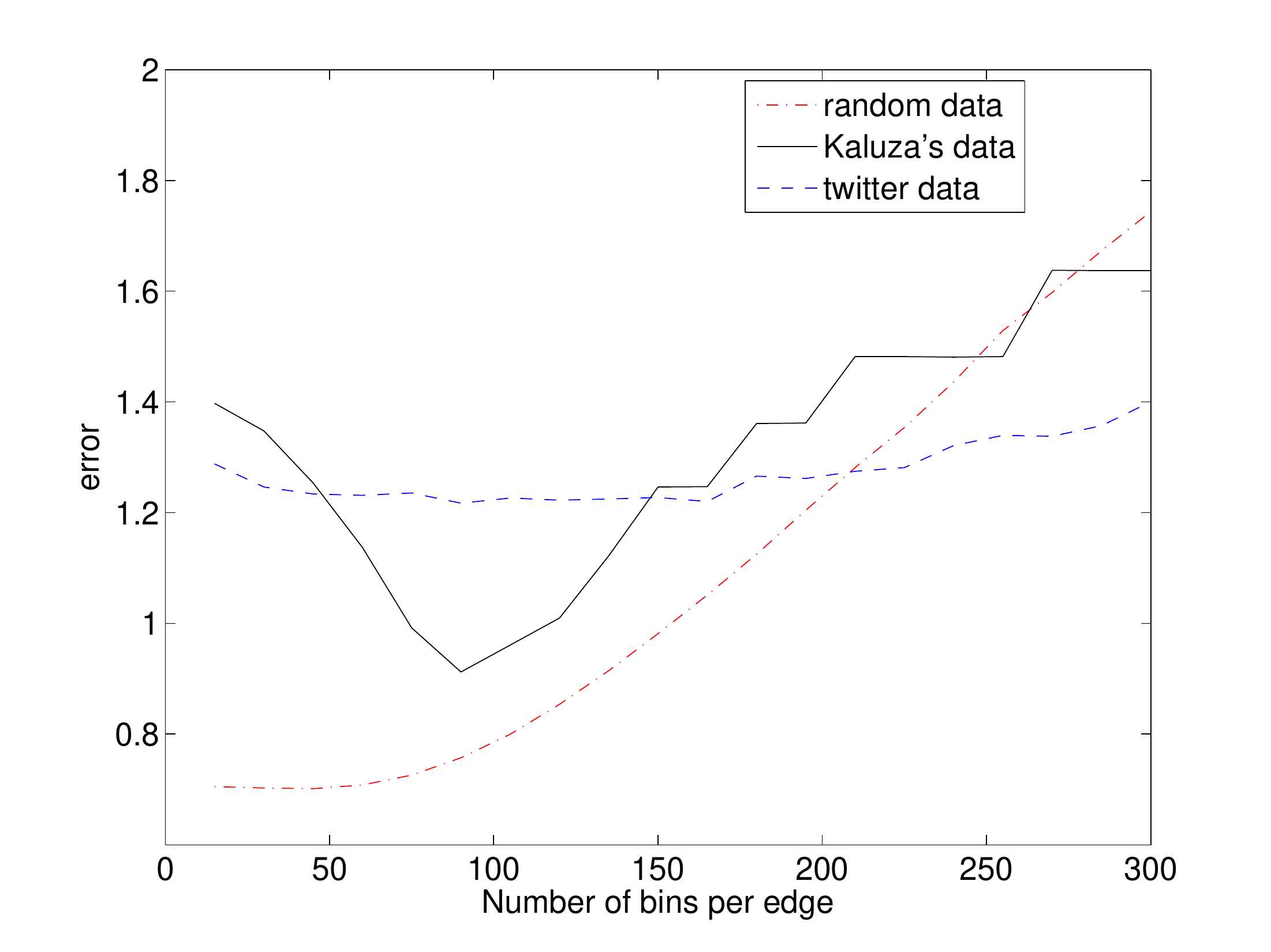}
\caption{The errors versus the bin size for different datasets. Each value on the graph is the average over 100 sample runs.}
\label{fig:histodistribution}
\end{figure}

\subsection{Range Query}

Many applications are interested to ask the total number of points within the query range in a dataset. It is desired to publish a noisy pointset meeting the privacy requirement, and yet able to provide accurate answers to the range queries.  Publishing an equi-width  histogram would not attain high accuracy if the size of the query ranges varies drastically.  Intuitively, wavelet-based techniques are  natural solutions to address such multi-scales queries. Xiao et al.~\cite{xiao2010differential} proposed a mechanism of adding Laplace noise to the coefficient  of a wavelet transformation of an equi-width histogram. The noisy wavelet coefficients are then published, from which range queries can be answered. Essentially, what being published is a series of equi-width histograms  with different widths (scales). Note that  there are quite a number of parameters to be determined prior to publishing, including the widths at varies scales and the amounts of privacy budget they consumed.


The answer to range queries can also be inferred from the output of our mechanism. Given a range,  we can estimate the number of points within the range from the estimated density function (as described in Section \ref{sec:density_function}) by  accumulating the probability over the query region, and then multiply by the total number of points.


We compare the wavelet-based mechanism, our mechanism and the equi-width histogram mechanism on the Twitter location dataset. For each range query, the absolute difference between the the true answer and the answer derived from the mechanism's output is taken as the error. We only consider square range queries in our experiments. For each query size $y$, 1,000 randomly selected square ranges with width $y$ are taken as the queries, and the average error is shown in Fig. \ref{fig:rangeqeury}.

In this experiment, we use Haar wavelet, and perform wavelet transform on the equi-width histogram with $ 512\times 512$ bins. After that,  appropriate noise is added to ensure $\epsilon$-differential privacy. To incorporate the knowledge of the database's size $n$, the DC component of the wavelet transform is set to be exactly $n$.
Under this setting, the best group size for our mechanism is 51.

Observe that as all mechanisms know the exact value of $n$, the accuracy improve when the range of the query covers more than half of the dataset. As expected,
the wavelet-base method outperforms the equi-width histogram mechanism in larger size range queries, but performs badly for small range due to the accumulation of noise. Surprisingly, our mechanism outperforms the equi-width histogram method for small range queries, and outperforms the wavelet based method for all sizes. This is possibly due to the fact that the locations of our queries are uniformly randomly chosen over a continuous domain, and thus, it is very likely that the query boundaries do not match the bins, leading to large error.



\begin{figure}[!h]
\includegraphics[width=0.53\textwidth]{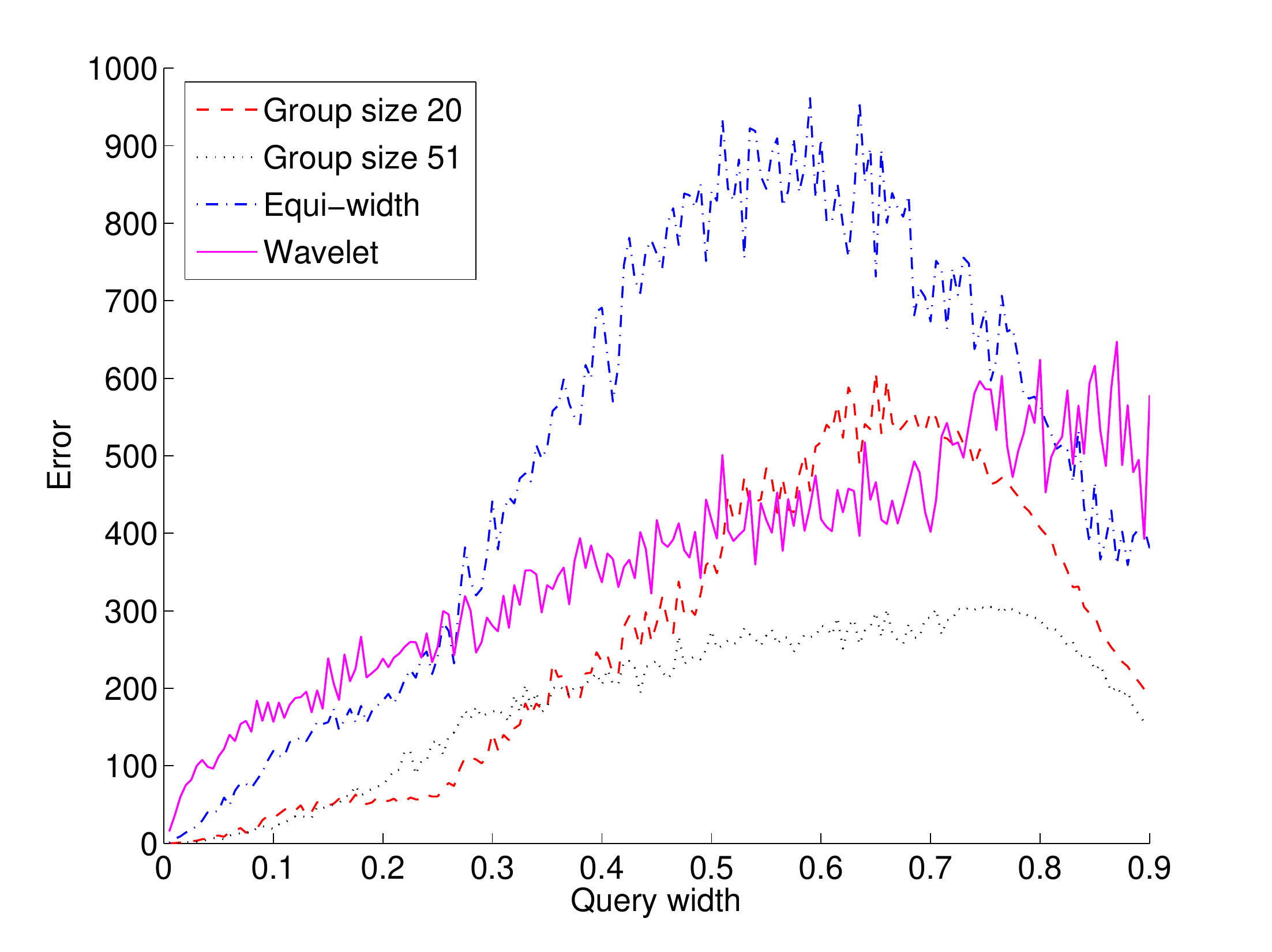}
\caption{The average range query error over 1000 random square range queries for each query size. The blue dash-dot line is the error of the equi-width histogram mechanism, the purple solid line is of the wavelet method, the red dashed line is our mechanism with group size 20 and the black dotted line is with group size 51.} \label{fig:rangeqeury}
\end{figure}

\subsection{Median}
Finding the median accurately in a differential private manner is challenging due to the high ``global sensitivity'': there are two datasets that differ by one element but having a completely different median. Nevertheless, for many instances, their ``local sensitivity'' are small. Nissim et al.~\cite{nissim2007smooth} showed that in general,  by adding noise proportional to the ``smooth sensitivity'' of the database instance, instead of the global sensitivity,  can also ensure differential privacy.  They also gave an $\Theta(n^2)$ algorithm that find the smooth sensitivity w.r.t. median.

Our mechanism  outputs the sorted sequence differential privately, it naturally gives the median. Compare to the smooth sensitivity-based mechanism, our mechanism can be efficiently carried out in $O(n)$ time with an sorted dataset.

We conduct experiments on  datasets of size  129 to compare the accuracy of both mechanisms.   Due to the quadratic running time in determining smooth sensitivity, we are unable to further investigate larger datasets.
The experiments are conducted for different local sensitivity and different $\epsilon$ values.
 To construct a dataset with different local sensitivity, 66 random numbers are generated with exponential distribution and then scaled to the unit interval. The dataset contains the 66 random numbers and 63 ones.
Figure \ref{fig:localsens} shows  the average noise level of both mechanisms on different local sensitivity, and Figure \ref{fig:epsilon} shows the noise level with different $\epsilon$ on a dataset that has a local sensitivity of $0.3$.

Observe that  when the local sensitivity of the median is high, our mechanism tends to provide a better result. In addition, our mechanism performs well under higher requirement of security: when the $\epsilon$ is smaller, the accuracy of our mechanism decreases slower than the smooth sensitivity-based method.

\begin{figure}[!h]
\includegraphics[width=0.53\textwidth]{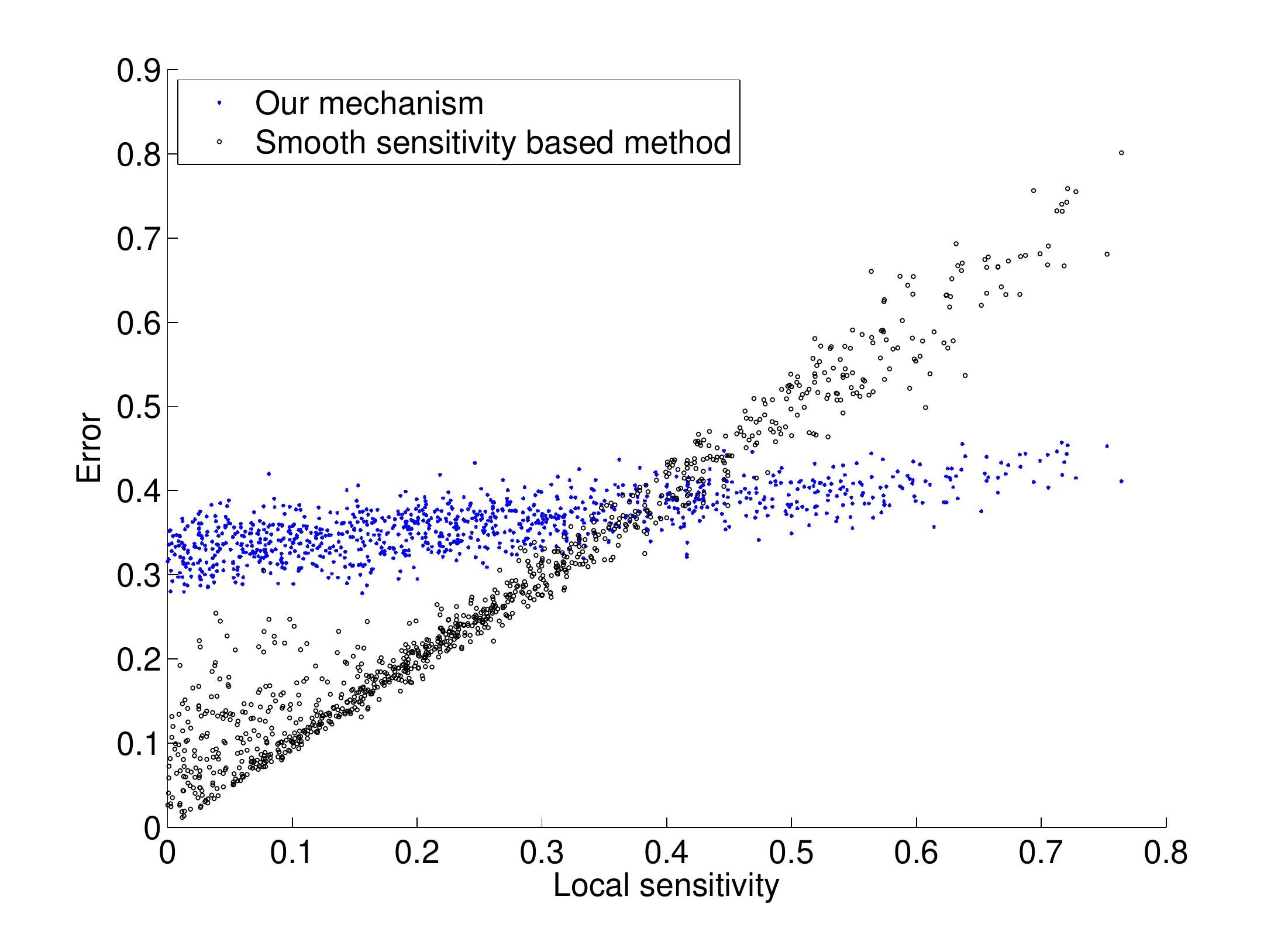}
\caption{The error of median output by the two mechanisms versus different local sensitivities. The blue dots are the error incurred by our mechanism and the black circles are the error incurred by the smooth sensitivity-based  mechanism. } \label{fig:localsens}
\end{figure}

\begin{figure}[!h]
\includegraphics[width=0.53\textwidth]{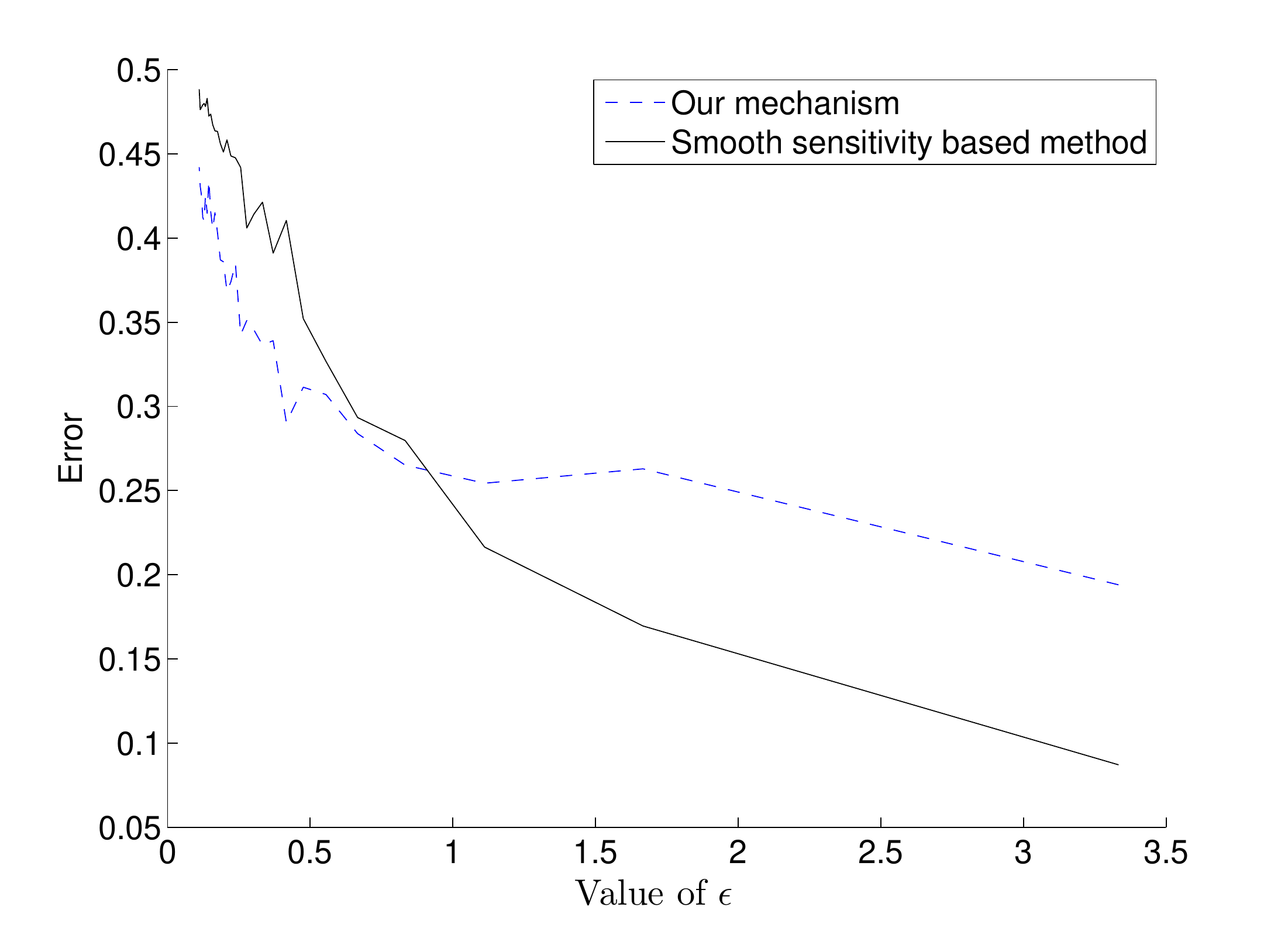}
\caption{The error of median versus different $\epsilon$. The blue dashed line and black solid line are the error incurred by our mechanism and smooth sensitivity-based  mechanism respectively. } \label{fig:epsilon}
\end{figure}

%


\section{Extension and Future Works}
\label{sec:extension}





\subsection{Hybrid Method}
The proposed mechanism can be viewed as the publishing of a ``fixed-depth'' histogram, where the number of elements in the histogram bins is fixed prior to publishing, whereas the mechanism outputs noisy bin's boundary. On the other hand, mechanisms based on frequency counts can be viewed as ``fixed-width'' histogram, where the boundary of the bins are fixed prior to publishing, whereas the mechanisms output noisy counts of elements in the bins.    The fixed-depth and fixed-width histogram could complement each other, by alternatively publishing one after another.  Here are two possibilities:

\subsubsection{Fixed-width-then-fixed-depth}
Let us take the Twitter location dataset  shown in Figure \ref{fig:NAorigindata} as an example. Observe that large portion of the region is sparse.  If the sparse region can be omitted, the sensitivity  of sorting would be significantly reduced.   This could be achieved by (1)first publishing an coarse equi-width  histogram with large width. (2) Next, for each bin, use the deterministic padding algorithm (Section \ref{sec:proposedapproach} Remark \ref{remark:2}) to extract $\widetilde{n}$ points, where $\widetilde{n}$ is the noisy count output by the equi-width histogram. (3) Finally, publish the extracted points  using our fixed-depth mechanism.  Note that the sensitivity for the fixed-depth mechanism is the width (or area) of the bin, which could be significantly smaller than the width (or area) of the  whole domain.
\subsubsection{Fixed-depth-then-fixed-width}
The unique solution of isotonic regression is a piecewise constant function. The steps in the solution lead to  artifacts of clustered data. It is interesting to investigate whether a subsequent fixed-width histogram could ``break'' the steps.

\subsection{Other  Metric}
It is interesting to investigate whether the proposed techniques can be applied to multidimensional data other than spatial data, for instance, tuples with attributes of age and gender.
\section{Related Work}
\label{sec:relatedwork}
There are extensive works on privacy-preserving data publishing.   The recent survey by Fung et al.~\cite{fung2010privacy}  gives a comprehensive overview on various notions, for example,  $k$-anonymity \cite{LSweeneyKAnonymity}, $\ell$-diversity \cite{lDiversity}, and differential privacy~\cite{dwork2006differential}.

Hay et al.\cite{hay2010boosting} proposed exploiting redundancies in the published data to boost accuracy, with supporting examples. One of the examples employs isotonic regression but in a way different from our mechanism.  They consider publishing  {\em unattributed histogram}, which is the (unordered) multiset of  the frequencies of a  histogram.  As the frequencies are unattributed (i.e. order of appearance is irrelevant), Hay et al. proposed publishing the sorted frequencies and later employing  isotonic regression to improve accuracy.   In contrast, our mechanism publishes the whole database.
It is no doubt that median is an important statistic. Finding median in a differentially private way is not easy due to the  large global sensitivity. Nissim et al.\cite{nissim2007smooth} introduced  the notion of smooth sensitivity and proposed an $\Theta(n^2)$ algorithm that computes the smooth sensitivity of an instance w.r.t. median.
Median has also  been used in the construction of other differential private mechanisms, for e.g.  dataset learning~\cite{blum2008learning} and  spatial decompositions~\cite{cormode2011differentially}.

\section{Conclusion and discussions}
\label{sec:conclusion}
%
%

Our mechanism is very simple from the publisher's point of view. The publisher just has to sort the points, group consecutive values, add Laplace noise and publish the noisy data. There is also minimal tuning to be carried out by the publisher. The main design decision is the choice of the group size $k$, which can be determined using our proposed noise models, and the locality preserving map which the classic Hilbert curve is suffice in attaining high accuracy.   Through empirical studies, we have shown that the published raw data contain rich information for the public to harvest, and provide high accuracy even for usages like median-finding, and range-searching that our mechanism is not initially designed for.  Such flexibility is desired for the need of ``{\em publish data, not the data mining result}'' as deliberated by Fung et al.~\cite{fung2010privacy}.

\bibliographystyle{abbrv}
\bibliography{diff}

\end{document}